\documentclass[11pt,reqno]{article}   

\usepackage[unicode=true]{hyperref}
\usepackage{amsmath}
\usepackage{cite}
\usepackage{amsfonts}
\usepackage{amssymb}
\usepackage{amsthm}
\usepackage{authblk}
\usepackage[pdftex]{color,graphicx}
\usepackage{adjustbox}
\usepackage{comment}
\usepackage{bbold}
\usepackage{tikz} 
\usetikzlibrary{decorations.pathreplacing,angles,quotes}
\usetikzlibrary{matrix,arrows,decorations.pathmorphing}
\usepackage{tikz-cd}
\usetikzlibrary{arrows}
\usetikzlibrary{decorations.markings}
\usepackage[all]{xypic}
\usepackage{bbold}
\usepackage{diagbox}
\usepackage{caption}
\usepackage{subcaption}
\usepackage{pdfpages} 
\usepackage{blkarray}
\usepackage{centernot}
\usepackage{mathtools}
\usepackage{stmaryrd}
\usepackage{soul}  
\definecolor{bluegray}{rgb}{0.4, 0.6, 0.8}
\definecolor{turquoise}{rgb}{0.2, 0.7, 0.6}

\newtheorem{thm}{Theorem}[section]
\newtheorem*{thm*}{Theorem}
\newtheorem{prop}[thm]{Proposition}
\newtheorem{lemma}[thm]{Lemma}

\newtheorem{corollary}[thm]{Corollary}

\newtheorem{subsec}[thm]{}
\newtheorem{defn}[thm]{Definition}
\newtheorem{example}[thm]{Example}

\newtheorem{remark}[thm]{Remark}

\newenvironment{myeq}[1][]
{\stepcounter{thm}\begin{equation}\tag{\thethm}{#1}}
	{\end{equation}}

\newtheorem{rem}[thm]{Remark}

\newcommand{\catConv}{\mathbf{Conv}}

\newcommand{\catC}{\mathbf{C}}
\newcommand{\catSet}{\mathbf{Set}}

\newcommand{\catCat}{\mathbf{Cat}}

\newcommand{\catsSet}{s\mathbf{Set}}

\newcommand{\adjoint}{\dashv}

\newcommand{\tightoverset}[2]{%
	\mathop{#2}\limits^{\vbox to -.5ex{\kern-0.25ex\hbox{$#1$}\vss}}}

\newcommand{\set}[1]{\ensuremath{ \lbrace #1 \rbrace }}

\newcommand{\Dec}{\text{Dec}}

\newcommand{\mysdg}[2][]{\myeq[#1]\xymatrix@R=25pt@C=15pt{#2}}
\newcommand{\myadiagnum}[2][]
{\stepcounter{thm}\begin{equation}
		\tag{\thethm}{#1}\vcenter{\xymatrix@R=20pt@C=25pt{#2}}\end{equation}}
\newcommand{\myadiagnumm}[2][]
{\stepcounter{thm}\begin{equation}
		\tag{\thethm}{#1}\vcenter{\xymatrix@R=15pt@C=32pt{#2}}\end{equation}}
\newcommand{\myaaadiag}[2][]
{\stepcounter{thm}\begin{equation}
		\tag{\thethm}{#1}\vcenter{\xymatrix@R=10pt@C=15pt{#2}}\end{equation}}
\newcommand{\mybbbdiag}[2][]
{\stepcounter{thm}\begin{equation}
		\tag{\thethm}{#1}\vcenter{\xymatrix@R=13pt@C=17pt{#2}}\end{equation}}
\newcommand{\mycccdiag}[2][]
{\stepcounter{thm}\begin{equation}
		\tag{\thethm}{#1}\vcenter{\xymatrix@R=15pt@C=30pt{#2}}\end{equation}}

\newcommand{\myrdiag}[2][]
{\stepcounter{thm}\begin{equation}
		\tag{\thethm}{#1}\vcenter{\xymatrix@R=20pt@C=30pt{#2}}\end{equation}}
\newcommand{\myssdg}[2][]
{\stepcounter{thm}\begin{equation}
		\tag{\thethm}{#1}\vcenter{\xymatrix@R=15pt@C=15pt{#2}}\end{equation}}
\newcommand{\mykkdiag}[2][]
{\stepcounter{thm}\begin{equation}
		\tag{\thethm}{#1}\vcenter{\xymatrix@R=2pt@C=13pt{#2}}\end{equation}}\newcommand{\mysdiag}[2][]
{\stepcounter{thm}\begin{equation}
		\tag{\thethm}{#1}\vcenter{\xymatrix@R=15pt@C=15pt{#2}}\end{equation}}
\newcommand{\mytdiag}[2][]
{\stepcounter{thm}\begin{equation}
		\tag{\thethm}{#1}\vcenter{\xymatrix@R=35pt@C=65pt{#2}}\end{equation}}
\newcommand{\myudiag}[2][]
{\stepcounter{thm}\begin{equation}
		\tag{\thethm}{#1}\vcenter{\xymatrix@R=19pt@C=19pt{#2}}\end{equation}}
\newcommand{\myvdiag}[2][]
{\stepcounter{thm}\begin{equation}
		\tag{\thethm}{#1}\vcenter{\xymatrix@R=35pt@C=50pt{#2}}\end{equation}}
\newcommand{\mydiagrm}[2][]
{\stepcounter{thm}\begin{equation}
		\tag{\thethm}{#1}\vcenter{\entrymodifiers={+++[o]}\xymatrix@R=25pt@C=25pt{#2}}\end{equation}}
\newcommand{\myrrdiag}[2][]
{\stepcounter{thm}\begin{equation}
		\tag{\thethm}{#1}\vcenter{\xymatrix@R=25pt@C=35pt{#2}}\end{equation}}
\newcommand{\myrrrdiag}[2][]
{\stepcounter{thm}\begin{equation}
		\tag{\thethm}{#1}\vcenter{\xymatrix@R=15pt@C=55pt{#2}}\end{equation}}
\newcommand{\NSbox}[4]{\ensuremath}
\newcommand{\stk}[1]{\stackrel{#1} {\longrightarrow}}

\newcommand{\wh}{\ -- \ }

\newcommand{\hsm}{\hspace{2 mm}}

\newcommand{\vsn}{\vspace{1 mm}}
\newcommand{\po}{\ar@{}[dr]|{\text{\pigpenfont R}}}
\newcommand{\pb}{\ar@{}[dr]|{\text{\pigpenfont J}}}

\newcommand{\lra}[1]{\langle{#1}\rangle}

\newcommand{\zz}{\mathbb{Z}}

\newcommand{\Vsupp}{\operatorname{Vsupp}}
\newcommand{\sDist}{\operatorname{sDist}}
\newcommand{\sseect}{\operatorname{sSect}}
\newcommand{\conv}{\operatorname{conv}}

\newcommand{\Id}{\operatorname{Id}}

\newcommand{\fs}{f\sb{\ast}}
\newcommand{\gs}{g\sb{\ast}}
\newcommand{\ks}{k\sb{\ast}}

\newcommand{\Xs}{X\sb{\ast}}
\newcommand{\Ys}{Y\sb{\ast}}
\newcommand{\Zs}{Z\sb{\ast}}

\newcommand{\Ws}{W\sb{\ast}}

\newcommand{\fhkg}{\raisebox{10pt}{\xymatrix@R=8pt@C=8pt{X
			\ar[r]^{h} \ar[d]_{f} & Z \ar[d]^{g} \\
			Y \ar[r]_{k} & W}}}
\newcommand{\fhkgs}{\xymatrix@R=8pt@C=8pt { \Xs \ar[r]^{h_{\ast}}
		\ar[d]_{\fs} & \Zs \ar[d]^{\gs} \\
		\Ys \ar[r]_{\ks} & \Ws } }
\begin{document}

\title{The geometry of simplicial distributions on suspension scenarios}


\author{Aziz Kharoof\footnote{aziz.kharoof@bilkent.edu.tr} }

\affil{{\small{Department of Mathematics, Bilkent University, Ankara, Turkey}}}




 \maketitle

%
%

\begin{abstract}
Quantum measurements often exhibit non-classical features, such as contextuality, which generalizes Bell's non-locality and serves as a resource in various quantum computation models. Existing frameworks have rigorously captured these phenomena, and recently, simplicial distributions have been introduced to deepen this understanding. The geometrical structure of simplicial distributions can be seen as a resource for applications in quantum information theory. In this work, we use topological foundations to study this geometrical structure, leveraging the fact that, in this simplicial framework, measurements and outcomes are represented as spaces.  This allows us to depict contextuality as a topological phenomenon. We show that applying the cone construction to the measurement space makes the corresponding non-signaling polytope equal to the join of $m$ copies of the original polytope, where $m$ is the number of possible outcomes per measurement. Then we glue two copies of cone measurement spaces to obtain a suspension measurement space. The decomposition done for simplicial distributions on a cone measurement space provides deeper insights into the 
geometry of simplicial distributions on a suspension measurement space and aids in characterizing the contextuality there. 
Additionally, we apply these results to derive a new type of Bell inequalities (inequalities that determine the set of local joint probabilities/non-contextual simplicial distributions) and to offer a mathematical explanation for certain contextual vertices from the literature.

\end{abstract}
\tableofcontents

\section{Introduction}
Probability distributions obtained from quantum measurements exhibit non-classical features. One such feature is contextuality, which generalizes Bell non-locality. Contextuality is recognized as a computational resource in various quantum computation schemes \cite{RaussContex,HowarContex,BravyiQuan,RauTherole}, and numerous frameworks have been developed to rigorously capture this concept. 
The sheaf-theoretic framework introduced by Abramsky and Brandenburger has proven to be particularly effective in formulating
contextuality \cite{Abrams1}. Any scenario can be described in this language and studied using advanced tools
such as cohomology \cite{Abrams2}. On the other hand, topological tools to study contextuality are first introduced in \cite{OkayTopo}. This formalism is based on chain complexes and applies to state-independent contextuality \wh a stronger form of contextuality.
Recently, an approach based on simplicial
distributions \cite{OkayQuan} was introduced, unifying the earlier two methods.
In this approach, simplicial distributions provide a new framework for studying contextuality, where measurements
and outcomes are represented by spaces instead of discrete sets of labels. 
The notion of a space is captured using combinatorial objects known as simplicial sets. These objects, while similar to simplicial complexes, are more expressive and form the foundational tools of modern homotopy theory \cite{JardineHomotopy}. The theory of simplicial distributions not only captures contextuality in an innovative way but also introduces new topological
ideas and tools for studying it and other fundamental principles in quantum foundations; 
see \cite{TopoloMeth,HomVert,OkayRank,EXtrCycle}. A simplicial set $X$ consists of a sequence of sets $X_n$, which represent the $n$-simplices, along with simplicial structure maps (most importantly, the face maps) that relate simplices of various dimensions. A simplicial distribution on a simplicial scenario, consisting of a measurement
space $X$ and an outcome space $Y$, is given by a simplicial set map
$$
p:X \to D(Y)
$$
where $D(Y)$ is the simplicial set whose simplices are distributions on the set of simplices of $Y$. More concretely,
a simplicial distribution $p$ consists of a family of distributions
$$
\set{p_{\sigma}:~ \text{$\sigma$ is an $n$-simplex in $X$ and $p_\sigma$ is a distribution on $Y_n$}}
$$ 
satisfying compatibility conditions induced by the simplicial structure relations. Typically, the outcome space is chosen to be $\Delta_{\zz_m}$ whose $n$-simplices are given by $n$-tuples of elements in $\zz_m$, and the face maps are given by deleting. This corresponds to the scenario in which each measurement has $m$ possible outcomes. Meanwhile, an $n$-simplex of $X$ represents a measurement context (an n-tuple of measurements that can be performed simultaneously).

%

The set of simplicial distributions on a simplicial scenario $(X,Y)$, denoted by $\sDist(X,Y)$, forms a polytope. Additionally, there exists a natural convex map:
$$
\Theta_{X,Y}: D(\catsSet(X,Y)) \to \sDist(X,Y)
$$
where $D(\catsSet(X,Y))$ is the set of distributions on the set of simplicial maps from $X$ to $Y$. The noncontextual 
simplicial distributions on $(X,Y)$ are defined to be the images of $\Theta_{X,Y}$, forming a subpolytope of $\sDist(X,Y)$. The facets of this subpolytope correspond to Bell inequalities \cite{BellFirst,Fineee}. Finding Bell inequalities and detecting extremal simplicial distributions \cite{PRBoxesss,Nonlocalcorrela} are of fundamental importance with applications to quantum computing. The goal of this paper is to characterize the geometric structure for the set of simplicial distributions, including the identification of extreme distributions and Bell inequalities, in scenarios where the measurement space is a cone space or a suspension space.
%

%
%
\begin{figure}[h!]
\centering
\begin{subfigure}{.33\textwidth}
  \centering
  \includegraphics[width=.6\linewidth]{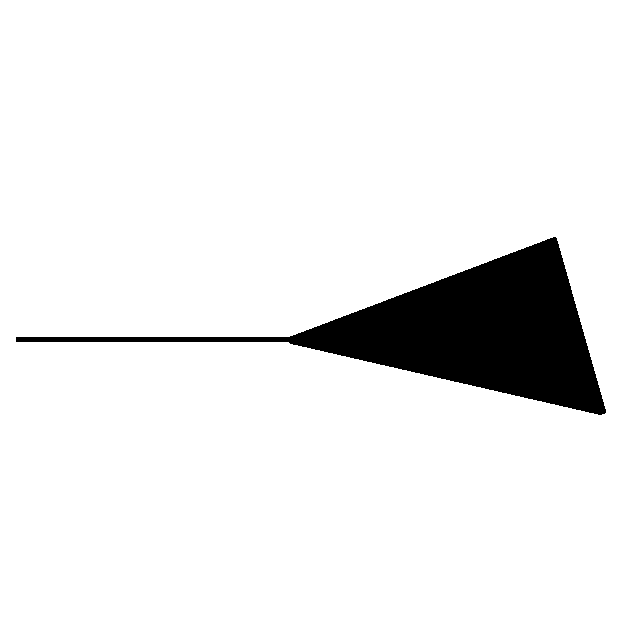}
  \caption{}
  \label{fig:MeasSpace}
\end{subfigure}%
\begin{subfigure}{.33\textwidth}
  \centering
  \includegraphics[width=.6\linewidth]{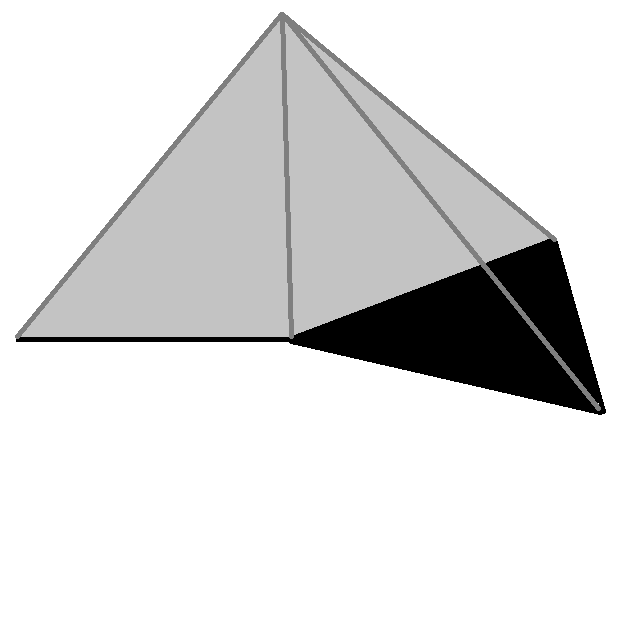}
  \caption{}
  \label{fig:Cone}
\end{subfigure}
\begin{subfigure}{.33\textwidth}
  \centering
  \includegraphics[width=.6\linewidth]{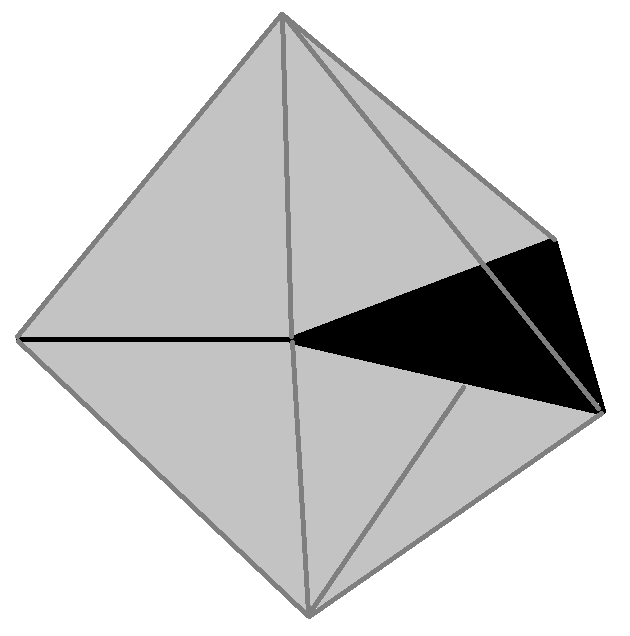}
  \caption{}
  \label{fig:Suspen}
\end{subfigure}
\caption{{(a) A measurement space $X$. (b) The cone of the space $X$. 
(c) The suspension of the space $X$.} 
}
\label{fig:mermin-scenario-and-os}
\end{figure}
In a typical Bell-type experiment, several observers perform measurements on a 
shared physical system. Each observer has a choice of different measurements to perform on his system.
Each measurement can yield some possible outcomes. A natural question arises: what happens if an additional party, with just one measurement, joins the experiment? 
In terms of the simplicial approach described above, this situation requires modifying the measurement space $X$ by adding a new vertex
and a set of simplices connecting each simplex in $X$ to the new vertex. This construction is known as the cone of $X$, denoted by $CX$ (see Figure \ref{fig:Cone}). Abstractly, this leads to a more general question: what is the relation between the simplicial distributions on $(CX,\Delta_{\zz_m})$ and those on $(X,\Delta_{\zz_m})$? The answer is  
\begin{equation}\label{eq:Introstar}
\sDist(CX,\Delta_{\zz_m})\cong \sDist(X,\Delta_{\zz_m})\star \dots \star \sDist(X,\Delta_{\zz_m})
\end{equation}
where the right-hand side in (\ref{eq:Introstar}) is the join of $m$ copies of $\sDist(X,\Delta_{\zz_m})$.
Moreover, we have:
\begin{thm}\label{Introthm:Decom}
Let $X$ be a connected simplicial set. A simplicial distribution $p$ on 
$(CX,\Delta_{\zz_m})$ can be uniquely decomposed as follows:  
\begin{equation}\label{eq:IntroDecomp}    
\left(\lra{\lambda_1,p^{(1)}},\dots,\lra{\lambda_m,p^{(m)}}
\right)    
\end{equation}
where $\sum_{i=1}^m\lambda_{i}=1$, $\lambda_{j} \in [0,1]$,
and $p^{(j)} \in \sDist(X,\Delta_{\zz_m})$ for every $1 \leq j \leq m$.
In addition, we have the following:
\begin{enumerate}
    \item $p$ is noncontextual if and only if
$p^{(j)}$ is noncontextual for every $1\leq j \leq m$.    
    \item $p$ is a vertex in $\sDist(CX,\Delta_{\zz_m})$ if
and only if there is $1\leq j \leq m$ such that $\lambda_j=1$ and $p^{(j)}$ is a vertex in $\sDist(X,\Delta_{\zz_m})$.
\end{enumerate}
\end{thm}
%
In homotopy theory, any map from a cone space is trivial (i.e., homotopic to a constant map). However, in certain higher homotopical constructions, such as Toda brackets \cite{TodC,KharoH}, interactions between maps from the cone can induce non-trivial maps from the suspension. The suspension $\Sigma X$ of a space $X$ is obtained by gluing two copies of cones $CX$ along $X$ (see Figure \ref{fig:Suspen}). The decomposition in (\ref{eq:IntroDecomp}) (for simplicial distributions on a cone measurement space) leads to the second main result concerning simplicial distributions on a suspension measurement space:
\begin{prop}
If $X$ is a connected simplicial set, then we have the following pullback square 
$$
\xymatrix@R =25pt@C=35pt {  
\sDist(\Sigma X,\Delta_{\zz_m}) \ar[r]\ar[d] & \sDist(X,\Delta_{\zz_m}) \star \dots
\star \sDist(X,\Delta_{\zz_m}) \ar[d]^{\Id\star \dots
\star \Id}\\
\sDist(X,\Delta_{\zz_m}) \star \dots
\star \sDist(X,\Delta_{\zz_m}) \ar[r]^<<<<<<<<<<<<<<<<<{\Id\star \dots
\star \Id} &  \sDist(X,\Delta_{\zz_m})
}
$$
In addition, a simplicial 
distribution 
$$
\left(\lra{\lambda_{1},p^{up,1}} ,\dots,\lra{\lambda_{m},p^{up,m}};\lra{\mu_{1},p^{down,1}},\dots,\lra{\mu_{m},p^{down,m}}\right)
$$
on $(\Sigma X ,\Delta_{\zz_m})$ is noncontextual if and only if for every $1\leq j \leq m$ there is 
$Q^{up,j},Q^{down,j} \in D(\catsSet(X,\Delta_{\zz_m}))$  such that 
\begin{enumerate}
    \item $\Theta(Q^{up,j})=p^{up,j}$ and $\Theta(Q^{down,j})=p^{down,j}$ for every $1\leq j \leq m$,
    \item $\lambda_{1}Q^{up,1}+\dots+\lambda_{m}Q^{up,m}
=\mu_{1}Q^{down,1}+\dots+\mu_{m}Q^{down,m}$.
\end{enumerate}
\end{prop}
Additional results proven in this paper are:
\begin{itemize}
    \item It is known that the set of distributions on a disjoint union is isomorphic (as a convex set) to the join of the sets of distributions on each individual set in the union (Proposition \ref{pro:LX1X2}). We generalize this fact to the set of distributions on a disjoint union parameterized by a connected space (Proposition \ref{pro:ConneDecom}).
%
    \item In Proposition \ref{pro:BellineqCone} we produce the Bell inequalities of the cone scenario $(CX,\Delta_{\zz_m})$ using the 
    Bell inequalities of the original scenario $(X,\Delta_{\zz_m})$ . This result is further applied to obtain a new type of Bell inequalities in Example \ref{ex:NewtypeBellineq}.
    \item   
    We demonstrate how specific families of vertices with particular properties in the scenario $(X,\Delta_{\zz_m})$ give rise to contextual vertices in the suspension scenario $(\Sigma X,\Delta_{\zz_m})$. The reason behind this lies in the restriction of the vertices to a line that included in $X$. As an application, we analyze the $(3,2,2)$ Bell scenario, i.e., three parties with two measurements per party and two outcomes per measurement. We provide a mathematical explanation for two classes of vertices appearing in the literature for this scenario.
\end{itemize}
The paper is organized as follows: In Section \ref{sec222}, we introduce the foundational material necessary for the study, including
simplicial sets, convex sets, simplicial distributions, and related tools. Section \ref{sec:DecomEvent} is about scenarios with event space that is equal to disjoint union of smaller components. We study simplicial distributions on such scenarios through those on scenarios with smaller components as their event space. This is used in section 
\ref{sec:Conescanarios} in order to prove our main theorem highlighted above. The transition to cone scenarios employs the D\' ecalage construction, introduced in section \ref{subsec:Decla}. In section \ref{sec:Suspentionnn}, we characterize 
noncontextual simplicial distributions on suspension scenarios. The section concludes with an analysis of two types of contextual vertices on suspension scenarios, which emerge from interactions between certain vertices on the original scenarios.







\paragraph{Acknowledgments.}
This work is supported by the Air Force Office of Scientific Research under award number FA9550-24-1-0257.



\section{Simplicial distributions}\label{sec222}
Simplicial distributions were first introduced in \cite{OkayQuan}. In this section, we briefly recapitulate the theory of simplicial distributions and provide some basic foundational tools. 
\subsection{Simplicial scenarios}
A simplicial distribution is defined for a space of measurements and a space of outcomes, both represented by simplicial sets. Simplicial sets are combinatorial models of topological spaces that are more expressive than simplicial complexes. A \emph{simplicial set} $X$ consists of a sequence of sets $X_0,X_1,\cdots$ and associated simplicial structure maps. 
These maps, given by face maps $d_i^X\colon X_n\to X_{n-1}$ and degeneracy maps $s_j^X\colon X_n\to X_{n+1}$. 
The elements in $X_n$ called $n$-simplices, so the face and degeneracy maps 
encode how to glue and collapse the simplices, respectively {(see, e.g.,\cite{FriedmansSett})}. We usually denote the simplicial structure maps by $d_i$ and $s_j$ omitting the simplicial set from the notation. 
A simplex is called 
\begin{itemize}
    \item \emph{degenerate} if it lies in the image of a degeneracy map,
    \item \emph{non-degenerate} if it is not degenerate, and
    \item A \emph{generating simplex} if it is non-degenerate and not a face of another simplex.
\end{itemize}
A \emph{simplicial map} $f\colon X\to Y$ between two simplicial sets is defined as a sequence of functions $\set{f_n\colon X_n\to Y_n}_{n\geq 0}$ 
that respect the face and the degeneracy maps. For a simplex $x\in X_n$ we will write $f_x$ instead of $f_n(x)$. With this notation the compatibility conditions are given by
$$
d_i(f_x)=f_{d_i(x)}  \;\; \text{and} \;\; s_j(f_x)=f_{s_j(x)}
$$
The category of simplicial sets is denoted by $\catsSet$.
\begin{defn}
The \emph{distribution monad} \cite[Section VI]{MaclaneCat} is a functor $D\colon  \catSet \to \catSet$ defined as follows:
\begin{itemize}
    \item For a set $X$, we define the set $D(X)$ of distributions on $X$ to be
$$
D(X)=\set{P\colon X \to [0,\infty):~ |\set{x\in X:~P(x)\neq 0}|<\infty \,\, \text{and} \,\,\sum_{x\in X}P(x)=1}    
$$
\item For a map $f\colon X\to Y$, the map $D(f)\colon  D(X) \to D(Y)$ defined by
$$
P\mapsto \left( y\mapsto \sum_{x\in f^{-1}(y)} P(x) \right). 
$$
\end{itemize}
The unit of the distribution monad $\delta_X\colon  X \xhookrightarrow{} D(X)$ sends $x\in X$ to the delta distribution
		$$
		\delta^x(x') = \left\lbrace
		\begin{array}{ll}
			1 & x'=x\\
			0 & \text{otherwise.}
		\end{array}
		\right.
		$$
\end{defn}
The distribution monad extended to a functor $D\colon \catsSet\to \catsSet$ by sending a simplicial set $X$
to the simplicial set $D(X)$ whose n-simplices are given by the set $D(X_n)$ and the simplicial structure maps are given by marginalization along the structure maps of $X$:
$$
D(d_i)\colon D(X_n) \to D(X_{n-1}) \;\; \text{and} \;\; D(s_i)\colon D(X_n) \to D(X_{n+1})
$$
It turns out that this extended functor is also a monad; see \cite[Proposition 2.4]{ConCat}. The unit of this monad satisfies $(\delta_X)_n=\delta_{X_n}$.

\begin{defn}\label{def:simpdist}
 A \emph{simplicial scenario} is a pair $(X,Y)$ of simplicial sets. A \emph{simplicial distribution} on this scenario is a simplicial map $p\colon X\to D(Y)$. We will write $\sDist(X,Y)$
for the set of simplicial distributions on $(X,Y)$. The simplicial sets $X$ and $Y$ are called the \emph{measurement space} and the \emph{outcome space}, respectively.
\end{defn}
\begin{defn}\label{def:deterdist}
A simplicial distribution of the form $\delta_Y\circ \varphi$ for some $\varphi \in \catsSet(X,Y)$ is called \emph{deterministic distribution} and denoted by $\delta^{\varphi}$.    
\end{defn}
\begin{defn}\label{defn:DeltaUU}
For a set $U$, let $\Delta_{U}$ be the simplicial set whose $n$-simplices are given by the
set $U^{n+1}$ and the simplicial structure maps are given by
$$
\begin{aligned}
d_i(a_0,a_1,\dots,a_n)&=(a_0,a_1,\dots,a_{i-1},a_{i+1},\dots,a_n) \\
s_j(a_0,a_1,\dots,a_n)&=(a_0,a_1,\dots a_{j-1},a_{j},a_{j},a_{j+1}\dots,a_n)
\end{aligned}
$$
\end{defn} 
For applications, our canonical choice for the outcome space will be $\Delta_{\zz_m}$, 
where $m \geq 2$. 
In this case, for a simplicial distribution $p \colon X \to D(\Delta_{\zz_m})$, and $\sigma \in X_n$, usually, we will write 
$p_{\sigma}^{a_0 a_1\dots a_n}$ instead of $p_{\sigma}(a_0,a_1,\dots,a_n)$.
Another useful outcome space is $N\zz_m$ the nerve of the group $\zz_m$; see for example \cite{JardineHomotopy}.

\begin{example}\label{ex:CHSHScenarioo}
A famous example, known as the CHSH scenario \cite{CHSHScennn}, can be described as a simplicial scenario. 
Consider a measurement space $X$ consisting of four generating $1$-simplices 
$\sigma_1,\sigma_2,\sigma_3,\sigma_4$, such that
\begin{equation}\label{eq:Howgluieddd}
d_0(\sigma_1)=d_1(\sigma_2)\;\; , \;\; d_0(\sigma_2)=d_1(\sigma_3)\;\; ,\;\;d_0(\sigma_3)=d_1(\sigma_4)\;\; , 
\;\; d_0(\sigma_4)=d_1(\sigma_1)
\end{equation}
These relations indicate that the simplices are cyclically connected, forming a loop-like structure. 
For a simplicial distribution $p \in \sDist(X,\Delta_{\zz_2})$ we have 
$$
p_{d_0(\sigma_k)}^a=D(d_0)(p_{\sigma_k})^a=\sum_{d_0(a_0,a_1)=a}p_{\sigma_k}^{a_0a_1}=
p_{\sigma_k}^{0a}+p_{\sigma_k}^{1a}
$$
for every $1\leq k \leq 4$ and $a \in \zz_2$. Similarly, $
p_{d_1(\sigma_k)}^a=
p_{\sigma_k}^{a0}+p_{\sigma_k}^{a1}
$. As a result, Equations (\ref{eq:Howgluieddd}) implies that 
$$
\begin{aligned}
p_{\sigma_1}^{00}+p_{\sigma_1}^{10} &=p_{\sigma_2}^{00}+p_{\sigma_2}^{01} \\
p_{\sigma_2}^{00}+p_{\sigma_2}^{10} &=p_{\sigma_3}^{00} +p_{\sigma_3}^{01} \\
p_{\sigma_3}^{00}+p_{\sigma_3}^{10} &=p_{\sigma_4}^{00}+p_{\sigma_4}^{01} \\
p_{\sigma_4}^{00}+p_{\sigma_4}^{10} &=p_{\sigma_1}^{00}+p_{\sigma_1}^{01}    
\end{aligned}
$$
\end{example}
A broader perspective on simplicial distributions is provided by the bundle approach, introduced in \cite{BundlePaper}. This framework extends the concept of simplicial scenarios by allowing measurements to take outcomes in different sets, rather than being restricted to a uniform outcome space. Another important advantage of working with bundle scenarios lies in the notion of morphisms between them.
\begin{defn}
A \emph{(simplicial) bundle scenario} is a simplicial map $f\colon E \to X$ such that $f_n$ is surjective for every $n \geq 0$. The simplicial sets $X$ and $E$ are called the \emph{measurement space} and the \emph{event space}, respectively.
A \emph{simplicial distribution} on $f$
is a simplicial map $p\colon X\to D(E)$ which makes the following diagram commutes
$$
\begin{tikzcd}[column sep=huge,row sep=large]
& D(E) \arrow[d,"{D(f)}"] \\
X \arrow[ur,"p"] \arrow[r,hook,"\delta_X"'] & D(X)
\end{tikzcd}
$$
We write $\sDist(f)$  
for the set of simplicial 
distributions on $f$.

\end{defn}
\begin{example}
Given a simplicial scenario $(X,Y)$. The projection map $f_{X,Y}\colon  
 X \times Y \to X$ is a bundle scenario.    
\end{example}

\begin{prop}\label{pro:sDist-characterization}  
(\!\!\cite[Proposition 4.9]{BundlePaper}) A simplicial map $p\colon  X\to D(E)$ belongs to $\sDist(f)$ if and only if
$
\set{e \in E_n~|~ p_n(x)(e) \neq 0} \subseteq f_n^{-1}(x) 
$
for every simplex $x\in X_n$.
\end{prop}
In this paper we will use morphisms between bundles scenarios of the form $\alpha\colon f \to g$, which is encoded by the following commutative diagram in $\catsSet$:  
%
\begin{equation}\label{dia:Moralphaa}
\begin{tikzcd}[column sep=huge,row sep=large]
E
\arrow[rr,"\alpha"]
\arrow[dr,"f"'] && E'
\arrow[dl,"g"] \\
&  X &  
\end{tikzcd}
\end{equation}

\begin{defn}\label{def:Moralphaa} 
The \emph{push-forward} of a simplicial distribution $p \in \sDist(f)$ along the morphism in (\ref{dia:Moralphaa}) is 
a simplicial distribution $\alpha_{\ast}(p) \in \sDist(g)$ which defined to be $\alpha_{\ast}(p)=D(\alpha)\circ p$. 
\end{defn}
The induced map  $\alpha_{\ast}\colon \sDist(f) \to \sDist(g)$ is a convex map. See \cite[Definition 4.12 and Proposition 5.3]{BundlePaper} for the general case.     
We will use the same notation when we have a map between outcome spaces. This means, 
a simplicial map $\alpha\colon Y \to Y'$ induces a convex map $\alpha_{\ast}
\colon \sDist(X,Y) \to \sDist(X,Y')$ which defined by setting $\alpha_{\ast}(p)=D(\alpha)\circ p$.
\subsection{Contextuality}\label{subsec:Conttt}
In this section we give the definition of contextuality within the simplicial framework, following \cite[Definition 3.10]{OkayQuan} and \cite[Definition 5.8]{BundlePaper}. In addition, we introduce the definition of Bell inequalities, which serve as a fundamental tool for characterizing contextuality.

There is a natural map
\begin{equation}\label{eq:Thetaaaa}
\Theta=\Theta_{X,Y}\colon  D(\catsSet(X,Y)) \to \sDist(X,Y)
\end{equation}
defined by sending $P\in D(\catsSet(X,Y))$ to the simplicial distribution $\Theta(P)$ that satisfies $
\Theta_{}(P)_x(y)=\sum_{\varphi_x=y} P(\varphi)$ for every $x \in X_n$ and $y\in Y_n$. 
\begin{defn}\label{def:Contextulityyyy}
A simplicial distribution $p\colon X \to D(Y)$ is called \emph{contextual} if it does not lie in the image of $\Theta_{X,Y}$.
Otherwise it is called \emph{noncontextual}.    
\end{defn}
Similarly, for a bundle scenario $f\colon E \to X$ there is a natural map
\begin{equation}\label{eq:Thetaaaaaf}
\Theta_{f}\colon  D(\sseect(f)) \to \sDist(f)
\end{equation}
where $\sseect(f)$ is the set of sections of $f$. 
A simplicial distribution $p\in \sDist(f)$ is called \emph{contextual} if it does not lie in the image of $\Theta_{f}$.
Otherwise it is called \emph{noncontextual}. The next proposition shows that the theory of simplicial distributions on simplicial scenarios is embedded in the theory of simplicial distributions on bundle scenarios.

\begin{prop}\label{pro:XYfXY}
Given a simplicial scenario $(X,Y)$ and let $f_{X,Y}\colon  
 X \times Y \to X$ be the projection map. We identify $p\colon  X \to D(Y)$ with $p'\colon X \to D(X\times Y)$ by setting 
 $p'_x(x,y)={p}_x(y)$, to obtain an isomorphism $\sDist(X,Y) \cong \sDist(f_{X,Y})$ in $\catConv$ such that 
 the notions of contextuality coincide.     
\end{prop}
\begin{proof}
See \cite[Remark 4.10]{BundlePaper}.     
\end{proof}
We generalize the usual notion of Bell inequality, which is defined
for Bell-type scenarios for non-locality, to apply to any
simplicial scenario (see \cite{Fineee,ContexFract}).
\begin{defn}\label{def:Bellineq} 
For simplices $x_1,\dots,x_k$ in $X$, 
simplices $y_1,\dots,y_k$ in $Y$, and 
real coefficients $B_1,\dots,B_k, R$, the inequality  
$$
B_1 p_{x_1}(y_1)+ \dots  +B_k p_{x_k}(y_k) \leq R
$$
is called \emph{Bell inequality} for the scenario $(X,Y)$ if 
\begin{itemize}
    \item It is satisfied by every noncontextual simplicial distribution $p$ on $(X,Y)$.
    \item It is saturated by some noncontextual simplicial distribution $p$ on $(X,Y)$.
    \item It is violated by some contextual simplicial distribution $p$ on $(X,Y)$.
\end{itemize}
A family of Bell inequalities that characterize noncontextuality are called \emph{the Bell inequalities of the scenario}.
\end{defn}
\begin{example}
It is well-known that a simplicial distribution $p$ 
on the CHSH scenario (Example \ref{ex:CHSHScenarioo}) is noncontextual if and only if it
satisfies the following Bell inequalities (which is called the CHSH inequalities):
\begin{equation}\label{eq:CHSHineq}
\begin{aligned}
0\leq p_{\sigma_1}^{00}+p_{\sigma_1}^{11}&+p_{\sigma_2}^{00}+p_{\sigma_2}^{11} + p_{\sigma_3}^{00}+p_{\sigma_3}^{11}-p_{\sigma_4}^{00}-p_{\sigma_4}^{11} \leq 2\\
0\leq p_{\sigma_1}^{00}+p_{\sigma_1}^{11}&+p_{\sigma_2}^{00}+p_{\sigma_2}^{11} - p_{\sigma_3}^{00}-p_{\sigma_3}^{11}+p_{\sigma_4}^{00}+p_{\sigma_4}^{11} \leq 2 \\  
0\leq p_{\sigma_1}^{00}+p_{\sigma_1}^{11}&-p_{\sigma_2}^{00}-p_{\sigma_2}^{11} + p_{\sigma_3}^{00}+p_{\sigma_3}^{11}+p_{\sigma_4}^{00}+p_{\sigma_4}^{11} \leq 2 \\
0\leq -p_{\sigma_1}^{00}-p_{\sigma_1}^{11}&+p_{\sigma_2}^{00}+p_{\sigma_2}^{11} + p_{\sigma_3}^{00}+p_{\sigma_3}^{11}+p_{\sigma_4}^{00}+p_{\sigma_4}^{11} \leq 2
\end{aligned}
\end{equation}
\end{example}
Other famous Bell inequalities are the Froissart inequalities \cite{Foresat}, which belong to  
the complete bipartite graph $K_{3,3}$ as a measurement space and $\Delta_{\zz_2}$ as an outcome space (see 
\cite[Equation (28)]{TopoloMeth}).

\subsection{Convexity and vertices}\label{subsec:ConvVert}
A \emph{convex set} consists of a set equipped with a ternary operation $\lra{-,-,-}\colon  [0,1]\times X \times X \to X$ satisfying specific axioms that generalize the notion of convex combinations (see \cite[Definition 3]{JacobsConv}). It is more convenient to use the notation $t x +(1-t)y$ instead
of $\lra{t,x,y}$ for $t\in[0,1]$ and $x,y \in X$. We define the set of convex combinations for 
$x_1,\dots,x_n \in X$ to be   
$$
\conv\set{x_1,\dots,x_n}=\set{t_1x_1+\dots+t_n x_n:\; \sum_{i=1}^n t_i=1}
$$
Another equivalent way to define convex set is as an algebra over the distribution monad $D$ (see \cite[Theorem 4]{JacobsConv}). We will denote the category of convex sets by $\catConv$. There is an adjunction
\begin{equation}\label{eq:Set adjunction Conv}
D:\catSet \adjoint \catConv :U
\end{equation}
where $U$ is the forgetful functor and $D$ sends a set $X$ to the free convex set $D(X)$.
\begin{prop}\label{pro:Thetatrans}
The natural map $\Theta_{X,Y}$ given in (\ref{eq:Thetaaaa}) is the transpose of 
$(\delta_{Y})_\ast \colon  \catsSet(X,Y) \to \sDist(X,Y)$  in
$\catConv$, with respect to the adjunction in (\ref{eq:Set adjunction Conv}). An analogous fact holds for the map $\Theta_f$ in   
 (\ref{eq:Thetaaaaaf}).
\end{prop}
\begin{proof}
See \cite[Proposition 2.16]{ConCat} and \cite[Section 5.2]{BundlePaper}.   
\end{proof}

\begin{defn}\label{def:vert}
Given a convex set $A$. An element $a \in A$ is called a \emph{vertex} (or \emph{extremal}) in $A$ if there is no distinct element $b,c \in A$  
and $0<t< 1$ such that $a=t b + (1-t)  c$.     
\end{defn}
For a simplicial scenario $(X,Y)$, every deterministic distribution $\delta^{\varphi}$ (see Definition \ref{def:deterdist}) is a vertex in $\sDist(X,Y)$ (see \cite[Proposition 5.14]{ConCat}). In fact, the deterministic distributions are the only noncontextual vertices. Here is an example of a contextual vertex. 
\begin{example}\label{ex:PRBoxxx}
The simplicial distribution 
$p:X \to D(\Delta_{\zz_2})$ on the CHSH scenario (Example \ref{ex:CHSHScenarioo}) defined by 
$p_{\sigma_1}=p_{+}$ and $p_{\sigma_{k}}=p_{-}$ for $2\leq k \leq 4$, where
\begin{equation}\label{eq:p+p-}    
p_{+}^{ab}=\left\lbrace
\begin{array}{cc}
1/2 & a+b=0\\
0 & a+b =1\\
\end{array}
\right.
\;\; ,\;\;
p_{-}^{ab}=\left\lbrace
\begin{array}{cc}
0 & a+b=0\\
1/2 & a+b =1\\
\end{array}
\right.
\end{equation}
is a contextual vertex called PR box \cite{PRBoxesss}. There are seven other PR boxes on the CHSH scenario. 
\end{example}
Now, we recall the definition of vertex support from \cite[Section 2.2]{EXtrCycle} as an important tool for studying the geometry of simplicial distributions.
\begin{defn}\label{def:Rrrelation}
{\rm 
Given two simplicial distributions $p,q\colon X \to D(Y)$ we write {$q\preceq p$} if for every $n \geq 0$, $x \in X_n$, and $y\in Y_n$, the inequality $q_{x}(y)\neq 0$ implies that $p_x(y)\neq 0$. 
}
\end{defn} 
\begin{defn}\label{def:Vsuppp}
{\rm
Given a simplicial distribution $p\colon X \to D(Y)$. The \emph{vertex support} $\Vsupp(p)$ is the set of vertices 
$q\in \sDist(X,Y)$
satisfying $q \preceq p$. 
}
\end{defn} 
\begin{defn}\label{def:closeVert}
A set $\set{q^{(1)},\dots,q^{(k)}}$ of vertices in $\sDist(X,Y)$ is called a \emph{closed set of vertices} if one of the following
equivalent conditions satisfied
\begin{itemize}
    \item $\Vsupp(\frac{1}{k}q^{(1)}+\dots+\frac{1}{k}q^{(k)})=\set{q^{(1)},\dots,q^{(k)}}$.
    \item $\Vsupp(\lambda_1 q^{(1)}+\dots+\lambda_k q^{(k)})\subseteq \set{q^{(1)},\dots,q^{(k)}}$ for every 
    $\lambda_1,\dots,\lambda_k \in [0,1]$ such that $\sum_{i=1}^k\lambda_i=1$.
\end{itemize}
\end{defn}
\begin{prop}
\label{thm:vertex on union} 
(\!\!\cite[Theorem 5.3]{EXtrCycle}) Given a finitely generated simplicial set $X$, such that 
$X=A\cup B$. A simplicial distribution $p\colon X\to D(\Delta_{\zz_m})$ is a vertex if and only if $p$ is the unique simplicial distribution in $\sDist(X,\Delta_{\zz_m})$ whose restrictions to $A$ and $B$ fall inside $\conv(\Vsupp(p|_A))$ and $\conv(\Vsupp(p|_B))$, respectively. 
\end{prop}

\subsection{The monoid structure on simplicial distributions}\label{sec:Monoidalstucture}
An additional algebraic feature in the theory of simplicial distributions is the
monoid structure on $\sDist(X,Y)$ when Y is a simplicial group (the face and degeneracy maps are group homomorphisms). In this section, we will restrict the discussion to the case when $Y=\Delta_{\zz_m}$. For further details, see \cite{ConCat}. 

\begin{defn}
Given a group $(G,+)$ and distributions $P$ and $Q$ on $G$. The \emph{convolution product} $P \ast Q \in D(G)$ is defined by
$$
P\ast Q(a) =\sum_{a_1+a_2=a} P(a_1)Q(a_2)
$$
where the sum runs over pairs $(a_1,a_2)\in G \times G$ such that $a_1+a_2=a$.
\end{defn}
The set of distributions $D(G)$ with convolution product is a convex monoid \cite[Section 4]{ConCat}.
Next, we give an important example which will be used in Section \ref{subsec:Finalll}.

\begin{example}\label{def:RRRRR}
We define $S \in D(\zz_m^2)$ to be 
$$
S(a,b)= \begin{cases}
\frac{1}{m} & \text{if}\;\; b=a+1 \\
0  & \text{otherwise,}
\end{cases}
$$    
%
%
and define $S^j$ to be the result of the convolution product of $S$ by itself $j$ times. More concretely 
\begin{equation}\label{eq:Sjjj}
S^j(a,b)= \begin{cases}
\frac{1}{m} & \text{if}\;\; b=a+j\\
0  & \text{otherwise.}
\end{cases}
\end{equation}
We will denote $S^m$ by $I$ and define $S^0$ to be $I$ as well.
In fact, $\set{I,S,\dots,S^{m-1}}$ form a cyclic group, which 
will be called the \emph{average group} in $D(\zz^2_m)$.
If $m=2$ then the average group is $\set{p_{+},p_{-}}$ (see Equations (\ref{eq:p+p-})). 
The group $\zz_m^2$ acts on the average group by setting $(c,d)\cdot S^j= \delta^{(c,d)} \ast S^{j}$ since we 
have
\begin{equation}\label{eq:deltaR}
\delta^{(c,d)} \ast S^{j}=S^{j+d-c}
\end{equation}
\end{example}
\begin{example}\label{ex:Uniform}
For a finite group $G$, the \emph{uniform distribution} $U \in D(G)$ defined by $U(a)=\frac{1}{|G|}$ for 
every $a\in G$. It has the following property:
$$
U \ast P=P \ast U =U
$$
for every $P \in D(G)$.
\end{example}
Since $\Delta_{\zz_m}$ is a simplicial group, the set of simplicial distributions $\sDist(X,\Delta_{\zz_m})$ is also a convex monoid \cite[Lemma 5.1]{ConCat}. The product defined as follows:
\begin{defn}
Given $p,q \in \sDist(X,\Delta_{\zz_m})$, the product $p \cdot q \in \sDist(X,\Delta_{\zz_m})$ is defined by
$
(p\cdot q)_{x} = p_x \ast q_x
$
for every $x\in X_n$.
\end{defn}
An important special case is when $p$ is a deterministic distribution 
$\delta^{\varphi}$. 
In this case, we have
%
$$
(\delta^{\varphi}\cdot q)_{x}(y) =\sum_{y_1+y_2=y} \delta^{{\varphi}_x}(y_1)q_x(y_2)
=\delta^{{\varphi}_x}({\varphi}_x)q_x(y-{\varphi}_x)=q_x(y-{\varphi}_x)   
$$
%
This produces an action of the group $\catsSet(X,\Delta_{\zz_m})$ on $\sDist(X,\Delta_{\zz_m})$. For that reason, usually we will write $\varphi\cdot q$ instead of $\delta^{\varphi}\cdot q$. 

\begin{lemma}\label{lem:produpreorderr}
For $\tilde{p}, p, \tilde{q}, q \in \sDist(X,\Delta_{\zz_m})$, if $\tilde{p} \preceq p$ and $\tilde{q} \preceq q$, then $\tilde{p} \cdot \tilde{q} \preceq p \cdot q$ (see Definition \ref{def:Rrrelation}).    
\end{lemma}
\begin{proof}
Suppose that $(\tilde{p} \cdot \tilde{q})_x^y\neq 0$ for some $x\in X_n$ and $y\in \zz_m^{n+1}$. Remember that
$(\tilde{p} \cdot \tilde{q})_x^y=\sum_{y_1+y_2=y}\tilde{p}_x^{y_1}\tilde{q}_x^{y_2}$, so there is $y'_1,y'_2 \in \zz_m^{n+1}$ such that 
$y'_1+y'_2=y$, $\tilde{p}_x^{y'_1}\neq 0$, and $\tilde{q}_x^{y'_2}\neq 0$. Since $\tilde{p} \preceq p$ and $\tilde{q} \preceq q$, we 
get that ${p}_x^{y'_1}\neq 0$ and ${q}_x^{y'_2}\neq 0$. Thus $({p} \cdot {q})_x^y=\sum_{y_1+y_2=y}{p}_x^{y_1}{q}_x^{y_2}\geq
{p}_x^{y'_1} {q}_x^{y'_2}\neq 0$.

\end{proof}
\begin{prop}\label{pr:psiVsupp}
For $\varphi \in \catsSet(X,\Delta_{\zz_m})$ and $p \in \sDist(X,\Delta_{\zz_m})$ we have 
$$
\Vsupp(\varphi \cdot p)=\varphi\cdot \Vsupp(p)
$$      
\end{prop}
\begin{proof}
Given a vertex $q$ in $\Vsupp(p)$. By \cite[part (3) of Proposition 2.5]{TopoloMeth} $\varphi\cdot p$ is a vertex. On the other hand, 
Lemma \ref{lem:produpreorderr} implies that  
$$
\varphi \cdot q =\delta^{\varphi} \cdot q\preceq \delta^{\varphi} \cdot p=\varphi \cdot p
$$
Thus $\varphi \cdot q\in \Vsupp(\varphi \cdot p)$. We proved that $\varphi\cdot \Vsupp(p) \subseteq \Vsupp(\varphi \cdot p)$. Let
$-\varphi$ be the inverse of $\varphi$ in $\catsSet(X,\Delta_{\zz_m})$, then we have
$$
\Vsupp(\varphi \cdot p) = \varphi \cdot (-\varphi) \cdot \Vsupp(\varphi \cdot p) \subseteq
\varphi \cdot \Vsupp\left(-\varphi \cdot (\varphi\cdot p)\right)
=\varphi\cdot \Vsupp(p)
$$
\end{proof}

%
%
%
%

\section{Decomposition of the event space}\label{sec:DecomEvent}

The join of convex sets serves as the coproduct in the category of convex sets, $\catConv$ \cite[Lemma 2.18]{OperadConv}. This section explores how decompositions of the event space or outcome space yield corresponding decompositions of simplicial distributions in terms of the join. The results are then applied to identify vertices of the associated convex polytope and analyze contextuality.

\subsection{The set of distributions on disjoint union}
The category of convex sets, being a category of algebras over a monad on $\catSet$, is cocomplete \cite[Proposition 9.3.4]{Cocompleteee}. We characterize finite coproduct of free convex sets as the set of distributions on disjoint union.  

\begin{defn}\label{def:bulletttt}
Let $\set{\bullet}$ be the one-element convex set. For any convex set $A$, we define the set 
$$
A_{\bullet}=\set{\lra{\lambda,a}\in [0,1]\times A \sqcup \set{\bullet}:\; \lambda=0 \;\; 
\text{iff} \;\; a=\bullet} 
$$
Sometimes we will write $\lra{0,x}$ even if $x$ is an expression that does not make sense. In this case, by $\lra{0,x}$ we mean $\lra{0,\bullet}$.
\end{defn}
The set $A_{\bullet}$ is a convex set with the structure:
$$
t_1\lra{\lambda_1,a_1}+\dots+t_m\lra{\lambda_m,a_m}=\lra{\zeta,\frac{t_1\lambda_1}{\zeta}a_1+
\dots+\frac{t_m\lambda_m}{\zeta}a_m}
$$
where $\zeta=t_1\lambda_1+\dots+t_m\lambda_m$. In addition, the construction in Definition \ref{def:bulletttt}
gives us an endfunctor on $\catConv$ such that for $f\colon A \to B$ in $\catConv$, the map 
$f_{\bullet}\colon  A_{\bullet} \to B_{\bullet}$ defined by $f_{\bullet}(\lambda,a)=(\lambda,f(a))$ (see \cite[Definition 3 and Lemma 4]{StateConv}).
Now, we define the coproduct of finite number of convex sets using \cite[Proposition 5]{StateConv}.
\begin{defn}
Given $A_1,\dots,A_m \in \catConv$, we define the coproduct (the \emph{join}) $A_1\star A_2 \star \dots \star A_m$ to be the following convex set: 
$$
\set{\left(\lra{\lambda_1,a_1},\dots,\lra{\lambda_m,a_m}\right) \in (A_1)_{\bullet}\times \dots \times(A_m)_{\bullet}:  \sum_{i=1}^m \lambda_i=1}
$$
with the structures maps $\kappa_i\colon A_i \to A_1\star A_2 \star \dots \star A_m$ defined by 
$\kappa_i(a)=\left(\lra{0,\bullet},\dots,\lra{1,a},\dots,\lra{0,\bullet}\right)$.
\end{defn}
\begin{prop}\label{pro:vertexAstar}
The vertices of $A_1\star A_2 \star \dots \star A_m$ are elements of the form $\kappa_i(a)$ where $a$ is a vertex in $A_i$.
\end{prop}
\begin{proof}
Every element $\left(\lra{\lambda_1,a_1},\dots,\lra{\lambda_m,a_m}\right) \in A_1\star A_2 \star \dots \star A_m$ can be written as follows:
$$
\begin{aligned}
&\lambda_1\left(\lra{1,a_1},\lra{0,a_2}\dots,\lra{0,a_m}\right) + \dots +
\lambda_m\left(\lra{0,a_1},\dots,\lra{0,a_{m-1}},\lra{1,a_m}\right) \\
&=\lambda_1\kappa_1(a_1) + \dots +
\lambda_m\kappa_m(a_m)
\end{aligned}
$$
Therefore, it suffices to show that $\kappa_i(a_i)$ is a vertex if and only if $a_i$ is 
a vertex. Since $\kappa_i$ is injective, if $\kappa_i(a_i)$ is a vertex then by \cite[Proposition 5.15]{ConCat} we get that $a_i$ is 
a vertex. Now, suppose conversely that for some $0<t<1$ we have
$$
\begin{aligned}
\kappa_i(a_i)&=t\left(\lra{\lambda_1,b_1},\dots,\lra{\lambda_m,b_m}\right)  +
(1-t)\left(\lra{\mu_1,c_1},\dots,\lra{\mu_m,c_{m}}\right) \\
&=\left(t\lra{\lambda_1,b_1}+(1-t)\lra{\mu_1,c_1},\dots,t\lra{\lambda_m,b_m}+(1-t)\lra{\mu_m,c_m}\right)\\
&=\left(\lra{\zeta_1,\frac{t\lambda_1}{\zeta_1}b_1+\frac{(1-t)\mu_1}{\zeta_1}c_1},\dots,
\lra{\zeta_m,\frac{t\lambda_m}{\zeta_m}b_m+\frac{(1-t)\mu_m}{\zeta_m}c_m}\right)
\end{aligned}
$$
where $\zeta_j=t\lambda_j+(1-t)\mu_j$ for $1 \leq j \leq m$. We conclude that 
$$
t\lambda_j+(1-t)\mu_j= 
\begin{cases}
1  & \text{if}\;\;  j =i \\
0  &  \text{otherwise.}
\end{cases}
$$
So $\lambda_i=\mu_i=1$ and $\lambda_j=\mu_j=0$ for every $j \neq i$. Thus $a_i=tb_i+(1-t)c_i$. 
If $a_i$ is a vertex then $b_i=c_i$. This means that  
$$
\begin{aligned}
\left(\lra{\lambda_1,b_1},\dots,\lra{\lambda_m,b_m}\right)
&=\left(\lra{0,b_1},\dots,\lra{1,b_i},\dots,\lra{0,b_{m}}\right)\\
&=\left(\lra{0,c_1},\dots,\lra{1,c_i},\dots,\lra{0,c_{m}}\right)\\
&=\left(\lra{\mu_1,c_1},\dots,\lra{\mu_m,c_{m}}\right)
\end{aligned}
$$
\end{proof}

The distribution monad as a left adjoint in (\ref{eq:Set adjunction Conv}) preserve colimits 
\cite[Theorem 4.5.3]{Riehlll}. In particular, we have the following: 
\begin{prop}\label{pro:LX1X2}
For $X_1,\dots,X_m \in \catSet$ we have a natural isomorphism
$$
\eta_{X_1,\dots,X_m}\colon  D(X_1 \sqcup  \dots \sqcup X_m) \to D(X_1) \star \dots \star D(X_m) 
$$
which is the inverse of the natural isomorphism  
$D(X_1) \star \dots \star D(X_m) \stk{\cong}  D(X_1 \sqcup  \dots \sqcup X_m)$
that induced by the inclusions $D(X_i) \xhookrightarrow{} D(X_1 \sqcup  \dots \sqcup X_m)$.
\end{prop}
To obtain the formula of the isomorphism $\eta_{X_1,\dots,X_m}$ in Proposition \ref{pro:LX1X2}, note that the compositions of the structure maps $X_i \xhookrightarrow{\delta_{X_i}} D(X_i)\stk{\kappa_i}  D(X_1) \star \dots \star D(X_m)$ induce a map 
\begin{equation}\label{eq:incluuuu}
X_1 \sqcup \dots \sqcup X_m \to D(X_1) \star \dots \star D(X_m)
\end{equation}
The map $\eta_{X_1,\dots,X_m}$ defined to be the transpose of the natural map in (\ref{eq:incluuuu}) under the adjunction in (\ref{eq:Set adjunction Conv}). One can check that for $P \in D(X_1 \sqcup  \dots \sqcup X_m)$ we have 
\begin{equation}\label{eq:LFormula}   
\eta_{X_1,\dots,X_m}(P)=\left(\lra{\sum_{x\in X_1} P(x),P^{(1)}},\dots,\lra{\sum_{x\in X_m} P(x),P^{(m)}}\right)
\end{equation}
where $P^{(i)} \in D(X_i)$ defined by $P^{(i)}(y)=\frac{P(y)}{\sum_{x\in X_i} P(x)}$. 

\subsection{Simplicial distributions with event space of disjoint union}

In this section, we extend the decomposition described in (\ref{eq:LFormula}) to simplicial distributions defined on connected measurement spaces.

\begin{defn}
Let $E$ be a simplicial set and $F \subseteq E$ a simplicial subset. We say that $F$ is a \emph{summand} of $E$ if 
there exists another simplicial subset $F'\subseteq E$ such that $E$ decomposes as a coproduct (disjoint union) $F \sqcup F'$. 
\end{defn}
%
%
\begin{lemma}\label{lem:pdx}
Given a bundle scenario $f\colon E \to X$ such that $F$ is a summand of $E$. Let $p$ be a simplicial distribution on $f$, then for 
every $x \in X_n$ we have:
\begin{enumerate}
    \item $\sum_{y\in F_n}p_x(y)=\sum_{y\in F_{n-1}}p_{d_i(x)}(y)$ for every $0\leq i \leq n$,
    \item  $\sum_{y\in F_n}p_x(y)=\sum_{y\in F_{n+1}}p_{s_i(x)}(y)$ for every $0\leq i \leq n$.
\end{enumerate}
\end{lemma}
\begin{proof}
Part $(1)$ follows from the following computation:
\begin{equation}\label{eq:sumdi=sum}
\begin{aligned}
\sum_{y\in F_{n-1}}p_{d_i(x)}(y)&=
\sum_{y\in F_{n-1}}D(d_i)(p_{x})(y)\\
&=
\sum_{y\in F_{n-1}} \sum_{d_i(y')=y}p_{x}(y')\\
&=\sum_{z\in F_n}p_x(z)
\end{aligned}
\end{equation}
We used the fact that $F$ is a summand of $E$ in the last step in (\ref{eq:sumdi=sum}). The proof of the second part is similar.
\end{proof}
\begin{defn}\label{def:lineee}
A \emph{line} is a simplicial set $L$ that generated by 
a sequence of pairwise distinct $1$-simplices $\sigma_1,\cdots,\sigma_n\in X_1$ satisfying
$$
d_{i_1}(\sigma_1)=d_{1-i_2}(\sigma_2)\, , \, 
d_{i_2}(\sigma_2)=d_{1-i_3}(\sigma_3)\,,\cdots,\, d_{i_{n-1}}(\sigma_{n-1})=d_{1-i_n}(\sigma_n)
$$ 
where $i_1,i_2, \cdots,i_n \in \{0,1\}$.
\end{defn}
\begin{defn}
A simplicial set $X$ is called \emph{connected} if for every vertices $u,v \in X_0$ there is a line $L \subseteq X$ as in
Definition \ref{def:lineee} such that $u=d_{1-i_1}(\sigma_1)$ and $v=d_{i_n}(\sigma_n)$. 
\end{defn}

\begin{lemma}\label{lem:sum=sum}
Given a bundle scenario $f\colon E \to X$ where $X$ is a connected simplicial set. Let $F$ be a summand of $E$ and let $p$ be a simplicial distribution on $f$. Then for every two simplices $x \in X_n$ and $x' \in X_k$, we have 
$$
\sum_{y\in F_n} p_x(y)=\sum_{y \in F_k} p_{x'}(y)
$$
\end{lemma}
\begin{proof}
Let $u=d_0^n(x)$ and $v=d_0^k(x')$ be the vertices obtained by applying $n$ times the face map $d_0$ on $x$ and $x'$, respectively. 
The space $X$ is connected, so there is a sequence of pairwise distinct $1$-simplices $\sigma_1,\cdots,\sigma_n\in X_1$ 
as in Definition \ref{def:lineee}
such that $u=d_{1-i_1}(\sigma_1)$ and $v=d_{i_n}(\sigma_n)$. Using 
part $(1)$ of Lemma \ref{lem:pdx} we obtain that  
$$
\begin{aligned}
\sum_{y\in F_0} p_u(y)&=\sum_{y \in F_1} p_{\sigma_1}(y)
=\sum_{y\in F_0} p_{d_{i_1}(\sigma_1)}(y)\\
&=\sum_{y\in F_0} p_{d_{1-i_2}(\sigma_2)}(y)
=\sum_{y \in F_1} p_{\sigma_2}(y)\\
&= \dots  =\sum_{y\in F_0} p_{d_{1-i_n}(\sigma_n)}(y)\\
&=\sum_{y \in F_1} p_{\sigma_n}(y) 
=\sum_{y\in F_0} p_v(y)
\end{aligned}
$$
Again, part $(1)$ of Lemma \ref{lem:pdx} implies that $\sum_{y\in F_n} p_x(y)=\sum_{y\in F_0} p_u(y)$ and $\sum_{y\in F_k} p_{x'}(y)=\sum_{y\in F_0} p_v(y)$.
\end{proof}

\begin{defn}\label{def:lambdpF}
Given a bundle scenario $f\colon E \to X$ where $X$ is a connected simplicial set. Let $F$ be a summand of $E$ and let $p$ be a simplicial distribution on $f$. We define 
\begin{itemize}
    \item The scalar $\lambda(p,F)$ to be $\sum_{y \in F_n} p_x(y)$ for some simplex $x \in X_n$. 
    \item The simplicial distribution $p|_{F} \in \sDist(f|_{F})$ by setting $(p|_{F})_x(y)=\frac{p_x(y)}{\lambda(p,F)}$ for every $x \in X_n$ and $y \in F_n$.
\end{itemize}
\end{defn}
By Lemma \ref{lem:sum=sum} the scalar $\lambda(p,F)$ is well-defined. We prove that $p|_F$ 
is well defined while we ignore the case that $\lambda(p,F)=0$. First, 
 we have 
$$
\sum_{y \in F_n}(p|_F)_x(y)=\sum_{y \in F_n} \frac{p_x(y)}{\lambda(p,F)} = \frac{\sum_{y \in F_n}p_x(y)}{\lambda(p,F)} =1
$$
Thus $(p|_F)_x \in D(F_n)$ for every simplex $x\in X_n$. 
Now, we prove that $p|_F$ is a simplicial map from $X$ to $D(F)$. Given $y\in F_{n-1}$, we have
$$
\begin{aligned}  
(p|_F)_{d_i(x)}(y)&=\frac{p_{d_i(x)}(y)}{\lambda(p,F)}
= \frac{D(d_i)(p_{x})(y)}{\lambda(p,F)}  \\
&=\frac{\sum_{z\in E_n \,\, \text{s.t} \,\, d_i(z)=y }p_x(z)}{\lambda(p,F)}  \\
&=\frac{\sum_{z\in F_n \,\, \text{s.t} \,\, d_i(z)=y }p_x(z)}{\lambda(p,F)} \\
&=\sum_{z\in F_n \,\, \text{s.t} \,\, d_i(z)=y }\frac{p_x(z)}{\lambda(p,F)} \\
&=\sum_{z\in F_n \,\, \text{s.t} \,\, d_i(z)=y }(p|_{F})_x(z)\\
&=D(d_i)\left((p|_{F})_x\right)(y)
\end{aligned}
$$
Similar argue works for the degeneracy maps. Finally, by the definition of $p|_{F}$ and Proposition \ref{pro:sDist-characterization} we conclude that $p|_{F} \in \sDist(f|_{F})$. 

\begin{prop}\label{pro:ConneDecom}
Given a bundle scenario $f\colon E \to X$ where $X$ is a connected simplicial set. Suppose that $E=\sqcup_{i=1}^m F^{(i)}$ and let $\alpha^{(j)}\colon  F^{(j)} \hookrightarrow E$ be the inclusion maps. Then the map 
$$
\eta=\eta_{F^{(1)},\dots,F^{(m)}}\colon \sDist(f) \to \sDist(f|_{F^{(1)}})\star \dots \star \sDist(f|_{F^{(m)}})
$$
$$
\;\;\;\;\;\; \;\;\;\;\; \;\;\;\;\;\;\; p \mapsto \left(\lra{\lambda(p,F^{(1)}),p|_{F^{(1)}}},\dots,\lra{\lambda(p,F^{(m)}),p|_{F^{(m)}}}
\right)
$$
is the inverse of the map $\alpha^{(1)}_\ast \star 
\dots \star \alpha^{(m)}_\ast$ in $\catConv$ (see Definition \ref{def:Moralphaa}).
\end{prop}
\begin{proof}
First note that 
$$
\sum_{i=1}^m\lambda(p,F^{(i)})=\sum_{i=1}^m \sum_{y\in F^{(i)}_n} p_x(y)=\sum_{y\in E_n}p_x(y)=1
$$
Thus $\left(\lra{\lambda(p,F^{(1)}),p|_{F^{(1)}}},\dots,\lra{\lambda(p,F^{(m)}),p|_{F^{(m)}}}\right) \in 
\sDist(f|_{F^{(1)}})\star \dots \star \sDist(f|_{F^{(m)}})$.  
Now, Given $q \in \sDist(f|_{F^{(j)}})$ and let us denote 
$\alpha^{(j)}_\ast(q)$ by $q'$, then $\lambda(q',F^{(i)})=\sum_{y\in F^{(i)}_n}q'_x(y)=0$ for every $i\neq j$. Therefore, we have  
$$
\begin{aligned}
\eta\left(\alpha^{(j)}_\ast(q)\right)&= \eta(q') 
=
(\lra{0,q'|_{F^{(1)}}},\dots,\lra{\lambda(q',F^{(j)})
,q'|_{F^{(j)}}}\dots,\lra{0
,q'|_{F^{(m)}}})\\
&=(\lra{0,\bullet},\dots,\lra{1
,q}\dots,\lra{0
,\bullet})  =\kappa_j(q)
\end{aligned}
$$
where $\kappa_j$ is the structure map from $\sDist(f|_{F^{(j)}})$ to $\sDist(f|_{F^{(1)}})\star \dots \star \sDist(f|_{F^{(m)}})$. We conclude that $\eta \circ (\alpha^{(1)}_\ast \star 
\dots \star \alpha^{(m)}_\ast) =\Id$. For the converse composition given $p \in \sDist(f)$, we have the following:
\begin{equation}\label{eq:alphastar}
\begin{aligned}
(\alpha^{(1)}_\ast \star 
\dots \star \alpha^{(m)}_\ast)\left(\eta(p)\right)&=(\alpha^{(1)}_\ast \star 
\dots \star \alpha^{(m)}_\ast)\left(\lra{\lambda(p,F^{(1)}),p|_{F^{(1)}}},\dots,\lra{\lambda(p,F^{(m)}),p|_{F^{(m)}}}\right)\\
&=\lambda(p,F^{(1)})\alpha^{(1)}_\ast(p|_{F^{(1)}})+\dots+\lambda(p,F^{(m)})\alpha^{(m)}_\ast(p|_{F^{(m)}})
\end{aligned}
\end{equation}
Given $x \in X_n$ and $y\in E_n$, then there is $j$ such that $y \in F^{(j)}_n$. As a result 
$\alpha^{(i)}_\ast(p|_{F^{(i)}})_x(y)=0$ for every $i \neq j$. Therefore, by Equation (\ref{eq:alphastar}) we obtain that
$$
\begin{aligned}
(\alpha^{(1)}_\ast \star 
\dots \star \alpha^{(m)}_\ast)\left(\eta(p)\right)_x(y)&=\lambda(p,F^{(j)})\alpha^{(j)}_\ast(p|_{F^{(j)}})_x(y) \\
&=\lambda(p,F^{(j)})(p|_{F^{(j)}})_x(y) \\
&=\lambda(p,F^{(j)})\frac{p_x(y)}{\lambda(p,F^{(j)})}=p_x(y)
\end{aligned}
$$ 
\end{proof}

\begin{defn}\label{def:lambdaXY}
Given a simplicial scenario $(X,Y)$ such that $X$ is connected. Let $Z$ be a summand of $Y$ and let $p$ be a simplicial distribution on $(X,Y)$. We define 
\begin{itemize}
    \item The scalar $\lambda(p,Z)$ to be $\sum_{y \in Z_n} p_x(y)$ for some simplex $x \in X_n$. 
    \item The simplicial distribution $p|_{Z} \in \sDist(X,Z)$ by setting $(p|_{Z})_x(y)=\frac{p_x(y)}{\lambda(p,Z)}$ for $x \in X_n$ and $y \in Z_n$.
\end{itemize}
\end{defn}
Under the assumptions of Definition \ref{def:lambdaXY} the space $X \times Z$ is a summand of $X \times Y$. So if $p' \in \sDist(f_{X,Y})$ is the corresponding simplicial distribution for $p$ as in Proposition \ref{pro:XYfXY}, then we have 
$$
\lambda(p,Z)=\lambda(p',X\times Z)  \;\; \text{and} \;\; (p|_{Z})'=p'|_{X\times Z}
$$
as in Definition \ref{def:lambdpF}.
\begin{corollary}\label{cor:ConneDecomXY}
Given a simplicial scenario $(X,Y)$ where $X$ is a connected simplicial set. Suppose that $Y=\sqcup_{i=1}^m Y^{(i)}$ and let $\alpha^{(j)}\colon  Y^{(j)} \hookrightarrow Y$ be the inclusion maps. Then the map 
$$
\eta_{Y^{(1)},\dots,Y^{(m)}}\colon \sDist(X,Y) \to \sDist(X,Y^{(1)})\star \dots \star \sDist(X,Y^{(m)})
$$
$$
\;\;\;\;\;\; \;\;\;\;\; \;\;\;\;\;\;\; p \mapsto \left(\lra{\lambda(p,Y^{(1)}),p|_{Y^{(1)}}},\dots,\lra{\lambda(p,Y^{(m)}),p|_{Y^{(m)}}}\right)
$$
is the inverse of the map $\alpha^{(1)}_\ast \star 
\dots \star \alpha^{(m)}_\ast$ in $\catConv$.
\end{corollary}
Corollary \ref{cor:ConneDecomXY} (as well Proposition \ref{pro:ConneDecom}) is not correct if $X$ is not connected. 
For example, if $X=\Delta^0 \sqcup \Delta^0$, where $\Delta^0$ is the standard $0$-simplex (the simplicial set that generated by one vertex), then we have
$$
\begin{aligned}
\sDist(\Delta^0 \sqcup \Delta^0,\Delta_{\zz_m} \sqcup \Delta_{\zz_m})&=
\catsSet(\Delta^0 \sqcup \Delta^0,D(\Delta_{\zz_m} \sqcup \Delta_{\zz_m})) \\ 
&\cong\catsSet(\Delta^0 ,D(\Delta_{\zz_m}\sqcup \Delta_{\zz_m})) \times \catsSet(\Delta^0 ,D(\Delta_{\zz_m}\sqcup \Delta_{\zz_m})) \\
&\cong D(\zz_m \sqcup \zz_m) \times D(\zz_m \sqcup \zz_m) \\
& \cong \left(D(\zz_m) \star D(\zz_m) \right) \times \left(D(\zz_m) \star D(\zz_m) \right)
\end{aligned}
$$
Meanwhile, 
$$
\begin{aligned}
\sDist(&\Delta^0 \sqcup \Delta^0,\Delta_{\zz_m}) \star \sDist(\Delta^0 \sqcup \Delta^0,\Delta_{\zz_m})
=\catsSet(\Delta^0 \sqcup \Delta^0,D(\Delta_{\zz_m})) \star \catsSet(\Delta^0 \sqcup \Delta^0,D(\Delta_{\zz_m}))\\
&\cong \left(\catsSet(\Delta^0 ,D(\Delta_{\zz_m}))\times \catsSet(\Delta^0 ,D(\Delta_{\zz_m})) \right) \star \left(\catsSet(\Delta^0 ,D(\Delta_{\zz_m}))\times \catsSet(\Delta^0 ,D(\Delta_{\zz_m})) \right)\\
& \cong \left(D(\zz_m)\times D(\zz_m)\right) \star \left(D(\zz_m)\times D(\zz_m)\right)
\end{aligned}
$$
\subsection{Vertices and contextuality with event space of disjoint union}
The decomposition described in Proposition \ref{pro:ConneDecom} is a powerful tool for simplifying the analysis of simplicial distributions. Here, we demonstrate how this decomposition aids in characterizing vertices and analyzing contextuality.
\begin{prop}\label{pro:vertexalphaq}
Given a bundle scenario $f\colon E \to X$ where $X$ is connected and 
$E=\sqcup_{i=1}^m F^{(i)}$. Let $\alpha^{(j)}\colon  F^{(j)} \hookrightarrow E$ be the inclusion maps. A simplicial 
distribution $p$ is a vertex in $\sDist(f)$ if and only if there exists $1\leq j \leq m$ and a vertex $q$ in 
$\sDist(f|_{F^{(j)}})$ such that $p=\alpha^{(j)}_\ast(q)$.
\end{prop}
\begin{proof}
By Proposition \ref{pro:ConneDecom} the simplicial distribution $p$ is a vertex in $\sDist(f)$ if and only if $\eta(p)$ is a vertex in 
$\sDist(f|_{F^{(1)}})\star \dots \star \sDist(f|_{F^{(m)}})$. According to Proposition 
\ref{pro:vertexAstar} the latter holds if and only there is $1 \leq j \leq m$ and a vertex $q$ in $\sDist(f|_{F^{(j)}})$ such that
$\eta(p)=\kappa_j(q)$. In this case
$$
p=(\alpha^{(1)}_\ast \star 
\dots \star \alpha^{(m)}_\ast )\left(\eta(p)\right)=(\alpha^{(1)}_\ast \star 
\dots \star \alpha^{(m)}_\ast)\left(\kappa_j(q)\right)= \alpha^{(j)}_\ast(q)
$$
\end{proof}

\begin{corollary}\label{cor:conn2}
Given a simplicial scenario $(X,Y)$ such that $X$ is connected and $Y=\sqcup_{i=1}^m Y^{(i)}$. 
Let $\alpha^{(j)}\colon  F^{(j)} \hookrightarrow E$ be the inclusion maps. A simplicial distribution $p$
is a vertex in $\sDist(X,Y)$ if and only if there exists $1\leq j \leq m$ and a vertex $q$ in $\sDist(X,Y^{(j)})$
such that $\alpha^{(j)}_{\ast}(q)= p$.  
\end{corollary}
%
%
Remember that for a bundle scenario $f\colon E \to X$ we write $\sseect(f)$ for the set of sections of $f$. Let $\alpha\colon f \to g$ be a morphism between bundle 
scenarios, if $\varphi\in \sseect(f)$ then 
$\alpha \circ \varphi \in \sseect(g)$. The induced map from $\sseect(f)$ to $\sseect(g)$ will also be denoted by $\alpha_{\ast}$.   
\begin{lemma}\label{lem:sectDecom}
Given a bundle scenario $f\colon E \to X$ where $X$ is connected. Suppose that $E=\sqcup_{i=1}^m F^{(i)}$ and let $\alpha^{(j)}\colon  F^{(j)} \hookrightarrow E$ be the inclusion maps. Then the map 
$$
D(\alpha^{(1)}_\ast) \star 
\dots \star D(\alpha^{(m)}_\ast)\colon  
D(\sseect(f|_{F^{(1)}}))\star \dots \star D\left(\sseect(f|_{F^{(m)}})\right) \to D\left(\sseect(f)\right)
$$ 
is invertible. 
\end{lemma}
\begin{proof}
Since the space $X$ is connected, the image of $\varphi \in \sseect(f)$ lies in one of the components 
$F^{(j)}$ of $E$. In other words, there is a unique 
$\psi \in \sseect(f|_{F^{(j)}})$ for some $1\leq j \leq m$ such that $\alpha^{(j)}_\ast(\psi)=\varphi$. We conclude that the map
$$
\alpha^{(1)}_\ast \sqcup \dots \sqcup \alpha^{(m)}_\ast \colon   \sseect(f|_{F^{(1)}}) \sqcup \dots \sqcup \sseect(f|_{F^{(m)}}) \to \sseect(f)
$$
is an isomorphism. On the other hand, by Proposition \ref{pro:LX1X2} we have an isomorphism 
\begin{equation}\label{eq:isomor1}
D(\sseect(f|_{F^{(1)}})) \star \dots \star D(\sseect(f|_{F^{(m)}})) 
\stk{\cong}  D\left(\sseect(f|_{F^{(1)}}) \sqcup  \dots \sqcup \sseect(f|_{F^{(m)}})\right)
\end{equation}
that induced by the inclusions $D\left(\sseect(f|_{F^{(i)}})\right) \xhookrightarrow{} D\left(\sseect(f|_{F^{(1)}}) \sqcup  \dots \sqcup \sseect(f|_{F^{(m)}})\right)$. Finally, the composition of the isomorphism in (\ref{eq:isomor1}) with the isomorphism
$D(\alpha^{(1)}_\ast \sqcup \dots \sqcup \alpha^{(m)}_\ast)$ is equal to the map
$D(\alpha^{(1)}_\ast) \star 
\dots \star D(\alpha^{(m)}_\ast)$.
\end{proof}
Now, we give the main result in this section.
\begin{prop}\label{theoremprop:Thm1}
Given a bundle scenario $f\colon E \to X$ where $X$ is connected. If $E=\sqcup_{i=1}^m F^{(i)}$, then the following diagram commutes
$$
\begin{tikzcd}[column sep=huge,row sep=large]
D\left(\sseect(f|_{F^{(1)}})\right)\star \dots \star D\left(\sseect(f|_{F^{(m)}})\right)
\arrow[d,"\cong"',"D(\alpha^{(1)}_\ast) \star 
\dots \star D(\alpha^{(m)}_\ast)"] \arrow[rr,"\Theta_{f|_{F^{(1)}}} \star \dots \star \Theta_{f|_{F^{(m)}}}"] && \sDist(f|_{F^{(1)}})\star \dots \star \sDist(f|_{F^{(m)}})\arrow[d,"\cong"',"\alpha^{(1)}_\ast \star 
\dots \star \alpha^{(m)}_\ast"]
\\
D(\sseect(f))
\arrow[rr,"\Theta_f"']
&&
\sDist(f)
\end{tikzcd}
$$
See Proposition \ref{pro:ConneDecom}.
\end{prop}
\begin{proof}
Given $\varphi \in \sseect(f|_{F^{(i)}})$. Because of the naturality of $\delta$ (the unit of the distribution monad) we get that  
$
D(\alpha^{(i)}) \circ \delta_{F^{(i)}}\circ \varphi = \delta_E \circ \alpha^{(i)}\circ \varphi 
$. Therefore, the following diagram commutes
\begin{equation}\label{dia:deltaalpha}
\begin{tikzcd}[column sep=huge,row sep=large]
\sseect(f|_{F^{(i)}}) \arrow[r,hook,"\left(\delta_{F^{(i)}}\right)_\ast"]
\arrow[d,hook,"\alpha^{(i)}_\ast"]
 & \sDist(f|_{F^{(i)}})  \arrow[d,hook,"\alpha^{(i)}_\ast"]  \\
\sseect(f)   \arrow[r,hook,"\left(\delta_{E}\right)_\ast"] &  \sDist(f)   
\end{tikzcd}
\end{equation}
Using Proposition \ref{pro:Thetatrans} the following commutative diagram induced by Diagram (\ref{dia:deltaalpha}) 
\begin{equation}\label{dia:Thetaalpha}
\begin{tikzcd}[column sep=huge,row sep=large]
D\left(\sseect(f|_{F^{(i)}})\right) \arrow[r,"\Theta_{f|_{F^{(i)}}}"]
\arrow[d,hook,"D(\alpha^{(i)}_\ast)"]
 & \sDist(f|_{F^{(i)}})  \arrow[d,hook,"\alpha^{(i)}_\ast"]  \\
D\left(\sseect(f)\right)   \arrow[r,"\Theta_f"] &  \sDist(f)   
\end{tikzcd}
\end{equation}
Finally, we get the result by Diagram (\ref{dia:Thetaalpha}), Proposition \ref{pro:ConneDecom}, and Lemma \ref{lem:sectDecom}.
\end{proof}
\begin{corollary}\label{cor:conn1}
Given a bundle scenario $f\colon E \to X$ where $X$ is connected. If $E=\sqcup_{i=1}^m F^{(i)}$, then 
a simplicial distribution $p \in \sDist(f)$ is noncontextual if and only if  $p|_{F^{(j)}}$ is noncontextual for every 
$1\leq j \leq m$ (see Definition \ref{def:lambdpF}).
\end{corollary}
\begin{proof}
By Proposition \ref{pro:ConneDecom} and Proposition \ref{theoremprop:Thm1}.    
\end{proof}
\begin{corollary}\label{cor:ThethaYj}
Given a simplicial scenario $(X,Y)$ where $X$ is connected. If $Y=\sqcup_{i=1}^m Y^{(i)}$ then
we have the following commutative diagram:
\begin{equation}\label{dia:ThetaYj}
\begin{tikzcd}[column sep=small,row sep=large]
D\left(\catsSet(X,Y^{(1)})\right)\star \dots \star D\left(\catsSet(X,Y^{(m)})\right) 
\arrow[r,"
"]  \arrow[d,"\cong"',"D(\alpha^{(1)}_\ast) \star 
\dots \star D(\alpha^{(m)}_\ast)"]
&
\sDist\left(X,Y^{(1)}\right)\star \dots \star \sDist\left(X,Y^{(m)}\right) \arrow[d,"\cong"',"\alpha^{(1)}_\ast \star 
\dots \star \alpha^{(m)}_\ast"]  \\
D\left(\catsSet(X,Y)\right) \arrow[r,"\Theta_{X,Y}"'] & \sDist(X,Y) 
\end{tikzcd}
\end{equation}
where the top map is $\Theta_{X,Y^{(1)}} \star \dots \star \Theta_{X,Y^{(m)}}$. See Corollary \ref{cor:ConneDecomXY}.

\end{corollary}

\begin{corollary}
Given a simplicial scenario $(X,Y)$ where $X$ is connected. If $Y=\sqcup_{i=1}^m Y^{(i)}$ then 
a simplicial distribution $p \in \sDist(X,Y)$ is noncontextual if and only if  $p|_{Y^{(j)}}$ is noncontextual for every 
$1\leq j \leq m$ (see Definition \ref{def:lambdaXY}).
\end{corollary}

\section{Simplicial distributions on cone scenarios}\label{sec:Conescanarios}
In this section, we apply the decomposition of simplicial distributions done in section 
\ref{sec:DecomEvent} to obtain a similar decomposition when the measurement space is a cone space. This analysis leverages the adjunction relationship between the cone functor and the d\' ecalage functor as well as several critical properties of the d\' ecalage functor. 

\subsection{The D\' ecalage functor}\label{subsec:Decla} 
In this section, we show that the d\' ecalage endfunctor on simplicial sets can be naturally restricted to an endfunctor on small categories.

\begin{defn}(\!\!\cite{Decle})
\label{def:decalage} 
{\rm
The \emph{d\' ecalage} of a simplicial set $X$ is the simplicial set $\Dec^0 (X)$ obtained by shifting the degrees of the simplices of $X$ down, i.e., $(\Dec^0 (X))_n = X_{n+1}$, and forgetting the first face and degeneracy maps. 
}
\end{defn}  
Definition \ref{def:decalage} induces a functor $\Dec^0\colon \catsSet \to \catsSet$ equipped with the natural projection map 
$
d_0\colon  \Dec^0(X) \to X
$.
\begin{rem}
In fact, the d\' ecalage of simplicial set is an augmented simplicial set
(A simplicial set with additional set $X_{-1}$ and additional face map $d \colon X_0 \to X_{-1}$ 
such that $d\circ d_0=d\circ d_1$). Anyway, we can compose the d\' ecalage functor with the forgetful functor to simplicial sets since any augmented simplicial set has an underlying unaugmented simplicial set found by forgetting $X_{-1}$.
\end{rem}
%
%

The outcome space $\Delta_{\zz_m}$ can be described in terms of the d\'ecalage. Explicitly, there is an isomorphism of simplicial sets
$
\Delta_{\zz_m} \xrightarrow{\cong} \Dec^0(N\zz_m)
$
defined in degree $n$ by sending $(a_0,a_1,\cdots,a_{n})$ to the tuple $(a_0,a_1-a_0,a_2-a_1,\dots,a_{n}-a_{n-1})$. 

\begin{defn}\label{def:DecCatccc}
For a category $\catC$, the category $\Dec^0 (\catC)$ consists of the following data:
\begin{itemize}
    \item The objects of $\Dec^0 (\catC)$ are the morphisms of $\catC$.
    \item A sequence $A\stk{f} B \stk{g} C$ in $\catC$ considered as a morphism 
    from $f$ to $g\circ f$ in $\Dec^0 (\catC)$. We denote this morphism by $(f,g)$.
   \item The composition of $(f,g)$ with $(g\circ f,h)$ defined to be 
   $(f,h\circ g)$. 
\end{itemize}
\end{defn}
The identity map in $\Dec^0(\catC)(f,f)$ is 
$(f,\Id)$. For associativity, given morphisms
$(f,g)$, $(g\circ f,h)$, and $(h\circ g \circ f,k)$. Then we have 
$$
\begin{aligned}
(h\circ g \circ f,k)\circ \left((g\circ f,h)\circ (f,g)\right)&= (h\circ g \circ f,k)\circ (f,h \circ g)\\
&=(f,k \circ h \circ g)\\
&=(g\circ f,k\circ h)\circ (f,g)\\
&=\left((h\circ g \circ f,k)\circ (g\circ f,h)\right)\circ (f,g)
\end{aligned}
$$
%
%
%
The construction in Definition \ref{def:DecCatccc} is functorial and if we restrict ourselves to small categories, then Definition \ref{def:DecCatccc} is a special case of Definition \ref{def:decalage} in the following sense:
\begin{prop}\label{pro:DecN=NDec}
Let $\catCat$ denote the category of small categories. We have a natural isomorphism from the composition $\catCat \stk{\Dec^0} \catCat \stk{N} \catsSet$ to  the composition $\catCat \stk{N} \catsSet \stk{\Dec^0} \catsSet$ given by identifying 
$$
\xymatrix@R=15pt @C=20pt{
\ar[rr]^{(f_1,f_2)} && \ar[rr]^>>>>>>>>>{(f_2 \circ f_1,f_3)} && \dots \dots\ar[rrrr]^>>>>>>>>>>>>>>>>>{(f_n \circ\dots \circ f_2\circ f_1,f_{n+1})} &&&& 
\, \in  \left(N \Dec^0(\catC)\right)_n }
$$
with
$$
\xymatrix@R=15pt @C=20pt{
\ar[r]^{f_1} & \ar[r]^{f_2} & 
\ar[r]^>>>>{f_3} & \dots \dots\ar[r]^>>>>{f_{n+1}} &
\, \in \left(\Dec^0(N\catC)\right)_n }
$$
for a small category $\catC$.     
\end{prop}
\begin{prop}
If $\catC$ is a groupoid, then    \
$$
|\Dec^0(\catC)(f,g)|= \begin{cases}
1 & \text{if $f$ and $g$ have the same domain}  \\
0  & \text{otherwise.}
\end{cases} 
$$
\end{prop}
\begin{example}\label{Ex: NCG=DeltaG}
Given a group $G$, then $\Dec^0(G)$ is the category with objects $x\in G$ and a unique morphism between any two objects. We will denote this category by $\catC_G$, and write $x\Rightarrow y$ for the unique morphism from $x$ to $y$. Note that the nerve of $\catC_G$ is isomorphic to $\Delta_G$ (see Definition \ref{defn:DeltaUU}).
\end{example}

\begin{prop}\label{pro:CCCG}
Given a group $G$, then $\Dec^0(\catC_G)\cong\sqcup_{x\in G} \catC_G$.     
\end{prop}
\begin{proof}
For every object $x \in G$ we have a subcategory $\catC_{G,x}$ of $\Dec^0(\catC_G)$ consists of objects of the form $x \Rightarrow y$ where $y \in G$, and morphisms of the form $(x \Rightarrow y,y\Rightarrow z)$ where $y,z \in G$. One can check that $\catC_{G,x}\cong \catC_{G}$ and $\Dec^0(\catC_G)=\sqcup_{x\in G} \catC_{G,x}$.    
\end{proof}

\begin{prop}
We have an isomorphism of simplicial sets
\begin{equation}\label{eq:isomophissss}
\gamma\colon  \Dec^0(\Delta_{\zz_m})\stk{\cong}   
\sqcup_{a \in \zz_m} \Delta_{\zz_m}
\end{equation}
which is defined by $\gamma_n(a_{0}, a_{1},\dots ,a_{n+1})=
\left(a_{0},(a_{1},\dots ,a_{n+1})\right)$
\end{prop}
\begin{proof}
By Example \ref{Ex: NCG=DeltaG}, Propositions \ref{pro:DecN=NDec} and \ref{pro:CCCG}, we have
\begin{equation}\label{eq:isomophissss2}
\begin{aligned}
\Dec^0(\Delta_{\zz_m}) &\cong \Dec^0(N\catC_{\zz_m})  
\cong N \Dec^0(\catC_{\zz_m}) \\ 
&\cong N(\sqcup_{a \in \zz_m} \catC_{\zz_m})
=\sqcup_{a\in \zz_m} N\catC_{\zz_m} \\
 &\cong \sqcup_{a \in \zz_m} \Delta_{\zz_m}
\end{aligned}
\end{equation}
The first isomorphism in Equation (\ref{eq:isomophissss2}) sends 
the $n$-simplex $(a_{0}, a_{1},\dots ,a_{n+1})$ in $\Dec^0(\Delta_{\zz_m})$
to 
$a_{0}\Rightarrow a_{1} \Rightarrow \dots \Rightarrow a_{n+1}$,
which is sent by the second isomorphism in Equation (\ref{eq:isomophissss2}) to  
\begin{equation}\label{eq:iso3}
\xymatrix@R=15pt @C=20pt{
\ar[rrr]^{(a_0\Rightarrow a_{1},a_{1}\Rightarrow a_{2})} &&& \ar[rrr]^>>>>>>>>>>>>>{(a_0 \Rightarrow a_{2},a_{2}\Rightarrow a_{3})} 
&&& \dots \dots\ar[rrrr]^>>>>>>>>>>>>>>>>>{(a_0 \Rightarrow a_{{n}},a_{{n}}
\Rightarrow a_{n+1})} &&&&   }
\end{equation}
The third isomorphism in (\ref{eq:isomophissss2}) sends the simplex in (\ref{eq:iso3}) to $(a_0,a_{1}\Rightarrow a_{2}\Rightarrow \dots \Rightarrow a_{n+1})$, which is finally sent to 
$\left(a_0,(a_{1}, a_{2}, \dots ,a_{n+1})\right) \in \left(\sqcup_{a \in \zz_m} \Delta_{\zz_m}\right)_{n}$
\end{proof}

\subsection{The geometry of simplicial distributions on cone scenarios}
In this section, we introduce the cone scenario $(CX,\Delta_{\zz_m})$ derived from a given scenario $(X,\Delta_{\zz_m})$. Next, we characterize the contextuality on the cone scenario $(CX,\Delta_{\zz_m})$ in terms of the contextuality on the original scenario $(X,\Delta_{\zz_m})$. Same thing is done for characterizing vertices on cone scenarios. Finally, we derive the Bell inequalities of the scenario $(CX,\Delta_{\zz_m})$ using the Bell inequalities of the scenario $(X,\Delta_{\zz_m})$.
\begin{defn}\label{def:cone}
{\rm
The \emph{cone} of a simplicial set $X$, denoted $CX$, is defined as a new simplicial set constructed as follows:  
\begin{itemize}
\item $(C X)_n = \set{c_n} \sqcup X_n \sqcup \left( \sqcup_{k+l=n-1} \set{c_k} \times X_l \right)$ where $c_n=(s_0)^n(c)$, 
 and $c$ is the cone point.
\item For $(c_k,\sigma)\in \set{c_k}\times X_l$ 
$$
d_i(c_k,\sigma) = \left\lbrace
\begin{array}{ll}
(c_{k-1},\sigma) & i\leq k\\
(c_k,d_{i-1-k}\sigma) & i>k,
\end{array}
\right.
$$
where $(c_{-1},\sigma)=\sigma$,
and
$$
s_j(c_k,\sigma) = \left\lbrace
\begin{array}{ll}
(c_{k+1},\sigma) & j\leq k\\
(c_k,s_{{j}-1-k}\sigma) & {j}>k.
\end{array}
\right.
$$
\item For simplices in $\set{c_n}$ and $X_n$, the face and the degeneracy maps act as in $\Delta^0$ and $X$, respectively. 
\end{itemize} 
}
\end{defn}

The cone construction is functorial, and there exists a natural inclusion $i\colon X \xhookrightarrow{} CX$. Moreover, for 
a connected simplicial set $X$, we have a natural bijection  
\begin{equation}\label{eq:Cone vs declage adjunction}
 \catsSet(X,\Dec^0 (Y))\stk{\cong} \catsSet(CX,Y)
\end{equation}
see \cite[Corollary 2.1]{Decle2}, that comes from a more general adjunction. This bijection is given by sending $\varphi\colon  X \to 
\Dec^0 (Y)$ to $\varphi'\colon  C X \to Y$ where  
\begin{equation}\label{eq:Transposeeee}
\varphi'_{n+1}(c,x)=\varphi_n(x)
\end{equation}
for every $x \in X_n$. Equation (\ref{eq:Transposeeee}) determines $\varphi'$ since $(c_k,x) \in C X $ is a degenerate simplex
for every $k \geq 1$ and $d_0(c,x)=x$.
The bijection in (\ref{eq:Cone vs declage adjunction}) induces  
the following commutative diagram in $\catConv$:
\begin{equation}\label{eq:Pro 2.23}
\begin{tikzcd}[column sep=huge,row sep=large]
D\left(\catsSet(X,\Dec^0(Y))\right) \arrow[rr,"\Theta_{X,\Dec^0(Y)}"] \arrow[d,"\cong"]
&& \sDist\left(X,\Dec^0(Y)\right)
\arrow[d,"\cong"] \\
D\left(\catsSet(C X,Y )\right) 
\arrow[rr,"\Theta_{C X,Y}"]
  && \sDist(C X,Y)
\end{tikzcd}
\end{equation}
See \cite[Proposition 2.23]{ConCat} and note that $D(\Dec^0(Y))=\Dec^0(D(Y))$.
\begin{prop}\label{pro:ThetaCXX}
For a connected measurement space $X$, the following commutative diagram holds in $\catConv$:
\begin{equation}\label{eq:TheataConeeeeee}
\begin{tikzcd}[column sep=small,row sep=large]
D\left(\catsSet(X,\Delta_{\zz_m})\right) \star \dots \star D\left(\catsSet(X,\Delta_{\zz_m})\right)
\arrow[r,"
"] \arrow[d,"\cong"] & \sDist(X,\Delta_{\zz_m})\star \dots \star\sDist(X,\Delta_{\zz_m})
 \arrow[d,"\cong","\beta^{(0)} \star 
\dots \star \beta^{(m-1)}"'] \\
 D\left(\catsSet(C X,\Delta_{\zz_m} )\right) 
\arrow[r,"\Theta_{C X,\Delta_{\zz_m}}"']
  & \sDist(C X,\Delta_{\zz_m}) 
\end{tikzcd}
\end{equation}
where the top map is $\Theta_{X,\Delta_{\zz_m}} \star \dots \star \Theta_{X,\Delta_{\zz_m}}$ and $\beta^{(j)}$ is defined as follows:
$$
\beta^{(j)}(p)_{(c,x)}(a_0,a_1,\dots,a_{n})
=\begin{cases}
p_{x}(a_1,\dots,a_{n})  & \text{if} \;\; a_0=j \\
0 & \text{otherwise}
\end{cases}
$$
for every $p \in \sDist(X,\Delta_{\zz_m})$, $x \in X_{n-1}$ and $(a_0,a_1,\dots,a_{n}) \in (\Delta_{\zz_m})_n$.
\end{prop}
%
%
\begin{proof}
The isomorphism $\gamma$ in (\ref{eq:isomophissss}) induces the following commutative diagram:
\begin{equation}\label{dia:gammainduces}
\begin{tikzcd}[column sep=huge,row sep=large]
D\left(\catsSet(X,\bigsqcup_{a \in \zz_m} \Delta_{\zz_m} )\right) 
\arrow[rr,"\Theta_{X,\bigsqcup_{a \in \zz_m}\Delta_{\zz_m}}"]
 \arrow[d,"\cong","D\left((\gamma^{-1})_\ast\right)"'] && \sDist( X,\bigsqcup_{a \in \zz_m} \Delta_{\zz_m}) \arrow[d,"\cong","(\gamma^{-1})_\ast"'] \\
D\left(\catsSet(X,\Dec^0(\Delta_{\zz_m}))\right) \arrow[rr,"\Theta_{X,\Dec^0(\Delta_{\zz_m})}"]
&& \sDist\left(X,\Dec^0(\Delta_{\zz_m})\right)
\end{tikzcd}
\end{equation}
Consider Diagram (\ref{dia:ThetaYj}) where the index runs from $0$ to $m-1$ and $Y^{(j)}=\Delta_{\zz_m}$. Composing it with Diagrams (\ref{dia:gammainduces}) and (\ref{eq:Pro 2.23}) yields 
Diagram (\ref{eq:TheataConeeeeee}). It remains to prove that for $p\in \sDist(X,\Delta_{\zz_m})$,
the image of the map $\left((\gamma^{-1})_\ast \circ \alpha^{(j)}_\ast\right)(p)$
under the isomorphism given in 
(\ref{eq:Cone vs declage adjunction})
is equal to $\beta^{(j)}(p)$. Given $x \in X_{n-1}$ and $(a_0,a_1,\dots,a_{n}) \in \zz_m^{n+1}$, we have 
$$
\begin{aligned}
\left((\gamma^{-1})_\ast \left(\alpha^{(j)}_\ast (p)\right)\right)'_{(c,x)}(a_0,a_1,\dots,a_{n})&=
\left((\gamma^{-1})_\ast \left(\alpha^{(j)}_\ast (p)\right)\right)_{x}(a_0,a_1,\dots,a_{n})\\
&=\left(\alpha^{(j)}_\ast (p)\right)_{x}\left(a_0,(a_1,\dots,a_{n})\right) \\
&=\begin{cases}
p_{x}(a_1,\dots,a_{n})  & \text{if} \;\; a_0=j \\
0 & \text{otherwise}
\end{cases} \\
& =\beta^{(j)}(p)_{(c,x)}(a_0,a_1,\dots,a_{n})
\end{aligned}
$$
see Equation (\ref{eq:Transposeeee}). 
\end{proof}
According to Proposition \ref{pro:ThetaCXX} every simplicial distribution $p \in \sDist(CX,\Delta_{\zz_m})$ can be written uniquely as follows:
\begin{equation}\label{eq:Decompositionn}
p=\left(\lra{\lambda_0,p^{(0)}},\dots,\lra{\lambda_{m-1},p^{(m-1)}}
\right)    
\end{equation}
where $\sum_{i=0}^{m-1}\lambda_{i}=1$, $\lambda_{j} \in [0,1]$ and $p^{(j)} \in \sDist(X,\Delta_{\zz_m})$ for every $0 \leq j \leq m-1$. In order to find formulas for $\lambda_j$ and $p^{(j)}$, let $\tilde{p}$ be the image of $p$ under the inverse of the isomorphism in (\ref{eq:Cone vs declage adjunction}). By applying $\left(\alpha^{(0)}_\ast \star 
\dots \star \alpha^{(m-1)}_\ast\right)^{-1}\circ\gamma_\ast$ (see the map in (\ref{eq:isomophissss}) and Corollary \ref{cor:ConneDecomXY})
on $\tilde{p}$ we get that
\begin{equation}\label{pjlambdajformula}
\lambda_j=\sum_{a_1,\dots,a_n \in \zz_m}p_{(c,x)}(j,a_1,\dots,a_n)
\;\; \text{and} \;\; p^{(j)}_x(a_1,\dots,a_n)=\frac{p_{(c,x)}(j,a_1,\dots,a_n)}{\lambda_j}
\end{equation}
where in the formula of $\lambda_j$, it doesn't matter which simplex $x$ is chosen as in Lemma \ref{lem:sum=sum}.

\begin{thm}\label{thm:Decom}
Let $X$ be a connected measurement space, and let $p$ be a simplicial distribution on $(CX,\Delta_{\zz_m})$ written as in Equation (\ref{eq:Decompositionn}). The following statements hold:
%
\begin{enumerate}
    \item $p$ is noncontextual if and only if
$p^{(j)}$ is noncontextual for every $0\leq j \leq m-1$.    
    \item $p$ is a vertex in $\sDist(CX,\Delta_{\zz_m})$ if
and only if there is $0\leq j \leq m-1$ such that $\lambda_j=1$ and $p^{(j)}$ is a vertex in $\sDist(X,\Delta_{\zz_m})$.
\end{enumerate}
\end{thm}
\begin{proof}
Part $(1)$ is obtained by the commutativity of Diagram (\ref{eq:TheataConeeeeee}).
For part $(2)$, by Proposition \ref{pro:vertexAstar} the simplicial distribution
$
\left(\lra{\lambda_0,p^{(0)}},\dots,
\lra{\lambda_{m-1},p^{(m-1)}}\right)
$ 
is a vertex in  $\sDist(X,\Delta_{\zz_m})\star \dots \star\sDist(X,\Delta_{\zz_m})$ if and only if 
there is $0\leq j \leq m-1$ such that $\lambda_j=1$ and $p^{(j)}$ is a vertex in $\sDist(X,\Delta_{\zz_m})$.
\end{proof}
Recall Definition \ref{def:Vsuppp} from section \ref{subsec:ConvVert}. The vertex support of a simplicial distribution on $(CX,\Delta_{\zz_m})$ can be computed using the decomposition in (\ref{eq:Decompositionn}). 

\begin{prop}\label{pro:VsuppDecompo}
Let $p \in \sDist(CX,\Delta_{\zz_m})$ be the simplicial distribution in (\ref{eq:Decompositionn}). Then we have the following:
$$
\Vsupp(p)= \bigsqcup_{j \;:\;\;\lambda_j \neq 0} \kappa_j
\left(\Vsupp(p^{(j)})\right)
$$
\end{prop}
\begin{proof}
By Proposition \ref{pro:vertexAstar} every vertex in $\sDist(CX,\Delta_{\zz_m})\cong \sDist(X,\Delta_{\zz_m}) \star \dots \star \sDist(X,\Delta_{\zz_m})$ is of the form $\kappa_j(q)$ for some vertex $q$ in $\sDist(X,\Delta_{\zz_m})$. In addition, $\kappa_j(q) \preceq p$ if and only if $q \preceq p^{(j)}$.
\end{proof}

Finally, we apply Theorem \ref{thm:Decom} to find the Bell inequalities (Definition \ref{def:Bellineq}) when the measurement space is a cone space.
\begin{prop}\label{pro:BellineqCone}
Let $X$ be a connected simplicial set. If the following inequalities are the Bell inequalities of the scenario $(X,\Delta_{\zz_m})$
\begin{equation}\label{eq:SetofBellineq}
\begin{aligned}
& B_{11} p_{x_{11}}^{y_{11}}+ \dots  +B_{1k_1} p_{x_{1k_1}}^{y_{1k_1}} \leq R_1 \\
& B_{21} p_{x_{21}}^{y_{21}}+ \dots  +B_{2k_2} p_{x_{2k_2}}^{y_{2k_2}} \leq R_2 \\
& \;\;\;\;\;\;\;\;\;\;\;\;\; \dots \;\; \dots \;\; \dots \\
& B_{s1} p_{x_{s1}}^{y_{s1}}+ \dots  +B_{sk_s} p_{x_{sk_s}}^{y_{sk_s}} \leq R_s,
\end{aligned}
\end{equation}
then the Bell inequalities of the scenario $(CX,\Delta_{\zz_m})$ are given by 
%
$$
\begin{aligned}
& B_{11} p_{(c,x_{11})}^{j,y_{11}}+ \dots  +B_{1k_1} p_{(c,x_{1k_1})}^{j,y_{1k_1}} \leq R_1\sum_{y}p_{(c,x_{11})}^{j,y}
, \;\; 0\leq j \leq m-1\\
& B_{21} p_{(c,x_{21})}^{j,y_{21}}+ \dots  +B_{2k_2} p_{(c,x_{2k_2})}^{j,y_{2k_2}} \leq R_2\sum_{y}p_{(c,x_{21})}^{j,y}
, \;\; 0\leq j \leq m-1\\
&  \;\;\;\;\;\;\;\;\;\;\;\;\;        \;\;\;\;\;\;\;\;\;\;\;\;\;\;\;\;\;\;         \;\;\;\;\;\;\;\;\;\;\;\;\; \dots \;\; \dots \;\; \dots \\
& B_{s1} p_{(c,x_{s1})}^{j,y_{s1}}+ \dots  +B_{sk_s} p_{(c,x_{sk_s})}^{j,y_{sk_s}} \leq R_s\sum_{y}p_{(c,x_{s1})}^{j,y}
, \;\; 0\leq j \leq m-1
\end{aligned}
$$
%
\end{prop} 
\begin{proof}
According to part $(1)$ of Theorem \ref{thm:Decom}, a simplicial distribution $p\colon CX \to D(\Delta_{\zz_m})$ is noncontextual if and only if 
$p^{(j)}$ is noncontextual for every $0\leq j\leq m-1$. Let 
$B_{1} p_{x_{1}}^{y_{1}}+ \dots  +B_{k} p_{x_{k}}^{y_{k}} \leq R$ be one of the inequalities of (\ref{eq:SetofBellineq}). The 
simplicial distribution $p^{(j)}$ satisfies this inequality, so the equations in (\ref{pjlambdajformula}) yields:
$$
B_{1} p_{(c,x_{1})}^{j,y_{1}}+ \dots  +B_{k} p_{(c,x_{k})}^{j,y_{k}} \leq R\sum_{y}p_{(c,x)}^{j,y}
$$
for some $x \in X_n$. To obtain the result we choose $x$ to be $x_1$. Note that if $\lambda_j=0$ then $p_{(c,x)}^{j,y}=0$ for every $x$ in $X$ 
and $y$ in $\Delta_{\zz_m}$. In this case, the inequality holds trivially. 
\end{proof}


%

\begin{example}\label{ex:NewtypeBellineq}
Let $X$ be the measurement space for the CHSH scenario described in Example \ref{ex:CHSHScenarioo}. 
Then $CX$ has four generating $2$-simplices 
$\tau_1,\tau_2,\tau_3,\tau_4$, such that
$$
d_1(\tau_1)=d_2(\tau_2)\;\; , \;\; d_1(\tau_2)=d_2(\tau_3)\;\; ,\;\;d_1(\tau_3)=d_2(\tau_4)\;\; , 
\;\; d_1(\tau_4)=d_2(\tau_1)
$$
Apply Proposition \ref{pro:BellineqCone} to the CHSH inequalities (\ref{eq:CHSHineq}) in order to obtain the Bell inequalities 
of the scenario $(CX,\Delta_{\zz_2})$
$$
\begin{aligned}
0\leq p_{\tau_1}^{000}+p_{\tau_1}^{011}&+p_{\tau_2}^{000}+p_{\tau_2}^{011} + p_{\tau_3}^{000}+p_{\tau_3}^{011}-
p_{\tau_4}^{000}-p_{\tau_4}^{011} \leq 2(p_{\tau_1}^{000}+p_{\tau_1}^{001}+p_{\tau_1}^{010}+p_{\tau_1}^{011})\\
0\leq p_{\tau_1}^{100}+p_{\tau_1}^{111}&+p_{\tau_2}^{100}+p_{\tau_2}^{111} + p_{\tau_3}^{100}+p_{\tau_3}^{111}-
p_{\tau_4}^{100}-p_{\tau_4}^{111} \leq 2(p_{\tau_1}^{100}+p_{\tau_1}^{101}+p_{\tau_1}^{110}+p_{\tau_1}^{111})\\
0\leq p_{\tau_1}^{000}+p_{\tau_1}^{011}&+p_{\tau_2}^{000}+p_{\tau_2}^{011} - p_{\tau_3}^{000}-p_{\tau_3}^{011}+p_{\tau_4}^{000}+p_{\tau_4}^{011} \leq 2(p_{\tau_1}^{000}+p_{\tau_1}^{001}+p_{\tau_1}^{010}+p_{\tau_1}^{011}) \\  
0\leq p_{\tau_1}^{100}+p_{\tau_1}^{111}&+p_{\tau_2}^{100}+p_{\tau_2}^{111} - p_{\tau_3}^{100}-p_{\tau_3}^{111}+p_{\tau_4}^{100}+p_{\tau_4}^{111} \leq 2(p_{\tau_1}^{100}+p_{\tau_1}^{101}+p_{\tau_1}^{110}+p_{\tau_1}^{111}) \\ 
0\leq p_{\tau_1}^{000}+p_{\tau_1}^{011}&-p_{\tau_2}^{000}-p_{\tau_2}^{011} + p_{\tau_3}^{000}+p_{\tau_3}^{011}+
p_{\tau_4}^{000}+p_{\tau_4}^{011} \leq 2(p_{\tau_1}^{000}+p_{\tau_1}^{001}+p_{\tau_1}^{010}+p_{\tau_1}^{011})\\
0\leq p_{\tau_1}^{100}+p_{\tau_1}^{111}&-p_{\tau_2}^{100}-p_{\tau_2}^{111} + p_{\tau_3}^{100}+p_{\tau_3}^{111}+
p_{\tau_4}^{100}+p_{\tau_4}^{111} \leq 2(p_{\tau_1}^{100}+p_{\tau_1}^{101}+p_{\tau_1}^{110}+p_{\tau_1}^{111})\\
0\leq -p_{\tau_1}^{000}-p_{\tau_1}^{011}&+p_{\tau_2}^{000}+p_{\tau_2}^{011} + p_{\tau_3}^{000}+p_{\tau_3}^{011}+
p_{\tau_4}^{000}+p_{\tau_4}^{011} \leq 2(p_{\tau_1}^{000}+p_{\tau_1}^{001}+p_{\tau_1}^{010}+p_{\tau_1}^{011})\\
0\leq -p_{\tau_1}^{100}-p_{\tau_1}^{111}&+p_{\tau_2}^{100}+p_{\tau_2}^{111} + p_{\tau_3}^{100}+p_{\tau_3}^{111}+
p_{\tau_4}^{100}+p_{\tau_4}^{111} \leq 2(p_{\tau_1}^{100}+p_{\tau_1}^{101}+p_{\tau_1}^{110}+p_{\tau_1}^{111})\\
\end{aligned}
$$
\end{example}

\section{Simiplicial distributions on suspension scenarios}\label{sec:Suspentionnn}

In this section, we use the characterization of simplicial distributions on cone scenarios, developed in the preceding section, to obtain a characterization of simplicial distributions on suspension scenarios. We then apply this characterization to detect two classes of contextual vertices on the suspension scenario, induced by vertices on the original scenario.

\subsection{Contextuality on suspension scenarios}
The suspension of a space $X$ is obtained by gluing two copies of the cone of $X$ along $X$. 
\begin{defn}\label{def:suspen}
The \emph{suspension} of $X\in\catsSet$, denoted by $\Sigma X$, is defined to be the following pushout in $\catsSet$:    
\begin{equation}\label{dia:suspen}
\begin{tikzcd}[column sep=huge,row sep=large]  
X \arrow[r,hook,"i_1"] \arrow[d,hook,"i_2"] & CX \arrow[d,hook,"s_1"] \\
CX \arrow[r,hook,"s_2"] & \Sigma X
\end{tikzcd}
\end{equation}
%
\end{defn}
\begin{prop}\label{pro:suspencontex}
Let $X$ be a connected measurement space. Then we have the following commutative diagram in $\catConv$:    
$$
\xymatrix@R =25pt@C=35pt {  
D\left(\catsSet(\Sigma X,\Delta_{\zz_m})\right)  \ar[d]_{} 
\ar[r]^{} \ar[rrd]^>>>>>>>>>>>{\Theta} & A \ar[d]^<<<{\Id\star \dots
\star \Id}
\ar[rrd]^{\Theta\star \dots
\star \Theta} & & \\
A \ar[r]^>>>>>>>>>>>>{\Id\star \dots
\star \Id} \ar[rrd]_{\Theta\star \dots
\star \Theta} & D\left(\catsSet( X,\Delta_{\zz_m})\right)   \ar[rrd]_<<<<<<<<<<<<<<<<{\Theta} & 
\sDist(\Sigma X ,\Delta_{\zz_m}) \ar[d]^<<<{} \ar[r]
&  B \ar[d]^{\Id\star \dots
\star \Id} \\
& & B \ar[r]_>>>>>>>>>>>>{\Id\star \dots
\star \Id} & \sDist(X ,\Delta_{\zz_m})  }
$$
where $A=D\left(\catsSet(X,\Delta_{\zz_m})\right) \star \dots \star D\left(\catsSet( X,\Delta_{\zz_m})\right)$, 
$B=\sDist(X,\Delta_{\zz_m}) \star \dots
\star \sDist(X,\Delta_{\zz_m})$, and the front square is a pullback square.
\end{prop}
\begin{proof}
Because of the naturality of $\Theta$, the pushout square in (\ref{dia:suspen}) induces the following commutative diagram:
%
$$
\xymatrix@R =25pt@C=35pt {  
D\left(\catsSet(\Sigma X,\Delta_{\zz_m})\right)  \ar[d]_{} 
\ar[r]^{} \ar[rrd]^>>>>>>>>>>>{\Theta_{\Sigma X,\Delta_{\zz_m}}} & D\left(\catsSet(CX,\Delta_{\zz_m})\right) \ar[d]^<<<{D(i_1^{\ast})}
\ar[rrd]^{\Theta_{CX,\Delta_{\zz_m}}} & & \\
D\left(\catsSet(CX,\Delta_{\zz_m})\right) \ar[r]^{D(i_2^{\ast})} \ar[rrd]_{\Theta_{CX,\Delta_{\zz_m}}} & D\left(\catsSet( X,\Delta_{\zz_m})\right)   \ar[rrd]_<<<<<<<<<<<<<<<<{{\Theta_{X,\Delta_{\zz_m}}}} & 
\sDist(\Sigma X ,\Delta_{\zz_m}) \ar[d]^<<<{} \ar[r]
&  \sDist(CX,\Delta_{\zz_m}) \ar[d]^{i_1^{\ast}} \\
& & \sDist(CX,\Delta_{\zz_m}) \ar[r]_{i_2^{\ast}} & \sDist(X ,\Delta_{\zz_m})  }
$$
where the front square is a pullback. Composing the two sides of this diagram 
with Diagram (\ref{eq:TheataConeeeeee}) yields: 
\begin{equation}\label{eq:TheataConeeCompose}
\xymatrix@R =25pt@C=35pt {
A \ar@/_{5.1pc}/[dd]_{\Id\star \dots
\star \Id}
\ar[rr]^{\Theta \star \dots \star \Theta} \ar[d]^{\cong} && B
 \ar[d]^{\cong}_{\beta^{(0)} \star 
\dots \star \beta^{(m-1)}} \ar@/^{4.5pc}/[dd]^{\Id\star \dots
\star \Id} \\
 D\left(\catsSet(C X,\Delta_{\zz_m} )\right) 
\ar[rr]^{\Theta_{C X,\Delta_{\zz_m}}} \ar[d]^{D(i^{\ast})}
  && \sDist(C X,\Delta_{\zz_m}) \ar[d]^{i^{\ast}}\\
  D\left(\catsSet(X,\Delta_{\zz_m} )\right) 
\ar[rr]^{\Theta_{X,\Delta_{\zz_m}}}
  && \sDist(X,\Delta_{\zz_m})
}
\end{equation}
%
%
Then we obtain the result using the above two diagrams.
\end{proof}
\begin{remark}\label{rem:distinsuspen}
By Proposition \ref{pro:suspencontex}, a simplicial distribution 
$p$ on $(\Sigma X,\Delta_{\zz_m})$ can be written uniquely as follows:
\begin{equation}\label{eq:DecomonSuspen}
p=\left(\lra{\lambda_{0},p^{up,0}} ,\dots,\lra{\lambda_{m-1},p^{up,m-1}};\lra{\mu_{0},p^{down,0}},\dots,\lra{\mu_{m-1},p^{down,m-1}}\right)
\end{equation}
where
\begin{itemize}
    \item $p^{up,j},p^{down,j} \in \sDist(X,\Delta_{\zz_m})$  \text{for every} $0\leq j \leq m-1$,
    \item  $\sum_{j=0}^{m-1}{\lambda_j}=1$ and $\sum_{j=0}^{m-1}{\mu_j}=1$,
    \item $\lambda_{0}p^{up,0}+\dots+\lambda_{m-1}p^{up,m-1}
=\mu_{0}p^{down,0}+\dots+\mu_{m-1}p^{down,m-1}$.
\end{itemize} 
\end{remark}
\begin{prop}
A simplicial distribution $p$ on $(\Sigma X,\Delta_{\zz_m})$ (as in Equation (\ref{eq:DecomonSuspen})) is noncontextual if and only if for every $0\leq j \leq m-1$ there is 
$Q^{up,j},Q^{down,j} \in D(\catsSet(X,\Delta_{\zz_m}))$  such that 
\begin{enumerate}
    \item $\Theta(Q^{up,j})=p^{up,j}$ and $\Theta(Q^{down,j})=p^{down,j}$ for every $0\leq j \leq m-1$,
    \item $\lambda_{0}Q^{up,0}+\dots+\lambda_{m-1}Q^{up,m-1}
=\mu_{0}Q^{down,0}+\dots+\mu_{m-1}Q^{down,m-1}$.
\end{enumerate}
\end{prop}
\begin{proof}
By \cite[Lemma 4.5]{OkayQuan} $p$ is noncontextual if and only if there is $Q,Q' \in D\left(\catsSet(CX,\Delta_{\zz_m})\right)$ 
such that $\Theta(Q)= s_1^{\ast}(p)$, $\Theta(Q')=s_2^{\ast}(p)$, and $D(i_1^\ast)(Q)=D(i_2^\ast)(Q')$ (see Diagram 
(\ref{dia:suspen})). Now, Diagram (\ref{eq:TheataConeeCompose}) implies that the mentioned above is equivalent to the existence of $Q^{up,j},Q^{down,j} \in D(\catsSet(X,\Delta_{\zz_m}))$ for 
every $0\leq j \leq m-1$ satisfying $(1)$ and $(2)$.
%
\end{proof}

\subsection{Complete collections of deterministic distributions}

In this section, we will write $L^{(3)}$ for the line with three edges $\sigma_1$, $\sigma_2$, and $\sigma_3$ (see Definition \ref{def:lineee}). For simplicity, we assume that $d_{0}(\sigma_1)=d_1(\sigma_2)$ and $d_{0}(\sigma_2)=d_1(\sigma_3)$.

\begin{defn}\label{def:ComplColle}
A \emph{complete collection of deterministic distributions} is 
a set $\set{\varphi^{i,j}:~0\leq i,j \leq m-1}$ of maps in $\catsSet(L^{(3)},\Delta_{\zz_m})$ such that 
\begin{enumerate}
    \item $\set{\varphi_{\sigma_1}^{i,j}:~  0\leq i,j \leq m-1}=\zz_m^2$.
    \item There is matrices $A=\begin{pmatrix} 
	a_{11} & a_{12}  \\
 a_{21} & a_{22} \\
	\end{pmatrix}$ and $B=\begin{pmatrix} 
	b_{11} & b_{12}  \\
 b_{21} & b_{22} \\
	\end{pmatrix}$ in $\zz_m^{2\times 2}$ such that $\varphi^{i,j}_{\sigma_2}=(i,j)A^T$ and 
    $\varphi^{i,j}_{\sigma_3}=(i,j)B^T$.
    \item There is $h\in \zz_m$ such that 
    $$
    gcd(a_{11}a_{22}-a_{12}a_{21},m)=gcd(b_{11}b_{22}-b_{12}b_{21},m)=gcd(a_{12}(b_{21}-b_{11}h)-a_{11}(b_{22}-b_{12}h),m)=1
    $$
\end{enumerate}
\end{defn}
\begin{example}\label{ex:compexam}
Let $m$ be an odd number. By setting 
$$
\varphi^{i,j}_{\sigma_1}=(i,i+j) \,\, ,\,\,\varphi^{i,j}_{\sigma_2}=(i+j,i+2j) \,\, , \,\, \varphi^{i,j}_{\sigma_3}=(i+2j,i+3j) 
$$
for $0\leq i,j \leq m-1$, we get a complete collection of deterministic distributions. In this case,  
$A=\begin{pmatrix} 
	1 & 1  \\
 1 & 2 \\
	\end{pmatrix}$ and $B=\begin{pmatrix} 
	1 & 2  \\
 1 & 3 \\
	\end{pmatrix}$,
so $a_{11}a_{22}-a_{12}a_{21}=1$, $b_{11}b_{22}-b_{12}b_{21}=1$, and $a_{12}b_{21}-a_{11}b_{22}=m-2$. 
\end{example}
For the next propositions, we recall the action of the group $\catsSet(X,\Delta_{\zz_m})$ on $\sDist(X,\Delta_{\zz_m})$ from Section \ref{sec:Monoidalstucture}. It is defined as follows:
$$
({\varphi}\cdot q )_{x}(y)=q_x(y-{\varphi}_x)   
$$
\begin{lemma}\label{lem:psih}
Given a complete collection of deterministic distributions $\set{\varphi^{i,j}:~0\leq i,j \leq m-1}$ as in Definition 
\ref{def:ComplColle}, and let 
$\psi^h\colon L^{(3)} \to \Delta_{\zz_m}$ be a simplicial map which defined by setting $\psi^h_{\sigma_1}=(0,0)$, $\psi^h_{\sigma_2}=(0,1)$, and 
$\psi^h_{\sigma_3}=(1,h)$. If for $\lambda_{ij},\mu_{ij} \in [0,1]$ with $\sum_{0\leq i,j \leq m-1}^m{\lambda_{ij}}=\sum_{0\leq i,j \leq m-1}^m{\mu_{ij}}=1$ we have
\begin{equation}\label{eq:psivarphi}
\sum_{0\leq i,j \leq m-1}\lambda_{ij}\delta^{\varphi^{i,j}}=\sum_{0\leq i,j \leq m-1}\mu_{ij}\psi^h \cdot \delta^{\varphi^{i,j}}   
\end{equation}
then 
$\lambda_{ij}=\mu_{ij}=\frac{1}{m^2}$ for every $0 \leq i,j \leq m-1$.
\end{lemma}
\begin{proof}
Since $\psi^h_{\sigma_1}=(0,0)$, by substituting $\sigma_1$ in Equation (\ref{eq:psivarphi}) we get that 
$$
\sum_{0\leq i,j \leq m-1}\lambda_{ij}\delta^{\varphi_{\sigma_1}^{i,j}}=
\sum_{0\leq i,j \leq m-1} \mu_{ij}\delta^{\varphi^{i,j}_{\sigma_1}}
$$ 
Therefore, using condition $(1)$ in Definition \ref{def:ComplColle} we conclude that 
\begin{equation}\label{eq:lamdamu}
\lambda_{ij}=\mu_{ij} \;\; \text{for every} \,\, 0\leq i,j \leq m-1
\end{equation}
Now, since $\psi^h_{\sigma_2}=(0,1)$, substituting $\sigma_2$ in Equation (\ref{eq:psivarphi}) yields
\begin{equation}\label{eq:sumsigma2}
\sum_{0\leq i,j \leq m-1}\lambda_{ij}\delta^{\varphi_{\sigma_2}^{i,j}}=
\sum_{0\leq i,j \leq m-1} \mu_{ij}\delta^{\varphi^{i,j}_{\sigma_2}+(0,1)}
\end{equation}
Remember that by condition $(2)$ in Definition \ref{def:ComplColle}, we have $\varphi^{i,j}_{\sigma_2}=(i,j)A^T$. By substituting this value 
in Equation (\ref{eq:sumsigma2}) we get that $\lambda_{ij}=\mu_{i'j'}$ for $(i',j')$ that satisfying  
$(i,j)A^T=(i',j')A^T+(0,1)$. Because $A$ is invertible we obtain 
\begin{equation}\label{eq:lambmuab}
\lambda_{ij}=\mu_{ij-(0,1)(A^{-1})^T}=\mu_{i+a,j+b} \;\; \text{where} \,\, (a,b)=(a_{11}a_{22}-a_{12}a_{21})^{-1}(a_{12},-a_{11})
\end{equation}
Similarly, by substituting $\sigma_3$ in Equation (\ref{eq:psivarphi}) we obtain 
\begin{equation}\label{eq:lambmucd}
\lambda_{ij}=\mu_{ij-(1,h)(B^{-1})^T}=\mu_{i+c,j+d} \;\; \text{where} \,\, 
(c,d)=(b_{11}b_{22}-b_{12}b_{21})^{-1}(hb_{12}-b_{22},b_{21}-b_{11}h)
\end{equation}
By Equations (\ref{eq:lamdamu}), (\ref{eq:lambmuab}), and (\ref{eq:lambmucd}) we conclude that 
$\lambda_{00}=\lambda_{ka+sc,kb+sd}$ for every $0 \leq k,s \leq m-1$. In addition, by condition $(3)$ in 
Definition \ref{def:ComplColle} the following number 
$$
det\begin{pmatrix} 
	a & b  \\
 c & d \\
	\end{pmatrix}
 =(a_{11}a_{22}-a_{12}a_{21})^{-1}(b_{11}b_{22}-b_{12}b_{21})^{-1}\left(a_{12}(b_{21}-b_{11}h)-a_{11}(b_{22}-hb_{12})\right) 
$$
is relatively prime to $m$. Thus $\lambda_{00}=\lambda_{ij}$ for every $0 \leq i,j \leq m-1$. Since $\sum_{0\leq i,j \leq m-1}{\lambda_{ij}}=1$, we obtain the desired result.
\end{proof}
\begin{prop}\label{pro:ConVertResult1}
Let $X$ be a finitely generated, connected simplicial set. Given a set of vertices $\set{q^{i,j}:~0\leq i,j \leq m-1}$ in $\sDist(X,\Delta_{\zz_m})$, and a line $L=L^{(3)} \subseteq X$ 
such that 
\begin{itemize}
    \item For every $0\leq j \leq m-1$ the set of vertices $\set{q^{i,j}:~0\leq i \leq m-1}$ is closed (see Definition \ref{def:closeVert}).
    \item $q^{i,j}|_{L}=\delta^{\varphi^{i,j}}$, where 
    $\set{\varphi^{i,j}:~0\leq i,j \leq m-1}$ is a complete collection of deterministic distributions as 
    in Definition \ref{def:ComplColle}.
    \item There is $\psi \in \catsSet(X,\Delta_{\zz_m})$ such that 
    $\sum_{0\leq i,j \leq m-1} \frac{1}{m^2} q^{i,j}=\psi \cdot \sum_{0\leq i,j \leq m-1} \frac{1}{m^2} q^{i,j}$ and
    $\psi|_{L}=\psi^h$ (see Lemma \ref{lem:psih}).
\end{itemize} 
Let $p^j=\frac{1}{m}q^{0,j}+\dots+\frac{1}{m}q^{m-1,j}$, then the simplicial distribution 
\begin{equation}\label{eq:FirstkindVert}   
\left(\lra{\frac{1}{m},p^{0}} ,\dots,\lra{\frac{1}{m},p^{m-1}};
\lra{\frac{1}{m},\psi\cdot p^{0}} ,\dots,\lra{\frac{1}{m},\psi\cdot p^{m-1}}\right)
\end{equation}
is a contextual vertex in $\sDist(\Sigma X ,\Delta_{\zz_m})$ (see Remark \ref{rem:distinsuspen}).
\end{prop}
\begin{proof}
First, the simplicial distribution in (\ref{eq:FirstkindVert}) is well defined because of the following: 
$$
\sum_{j=0}^{m-1} \frac{1}{m}p^j=\sum_{0\leq i,j \leq m-1} \frac{1}{m^2} q^{i,j}=\psi \cdot \sum_{0\leq i,j \leq m-1} \frac{1}{m^2} q^{i,j}=\psi \cdot \sum_{j=0}^{m-1} \frac{1}{m} p^j= \sum_{j=0}^{m-1} \frac{1}{m} \psi \cdot p^j
$$
Given simplicial distributions 
$$
p' \in \conv\left(\Vsupp(\lra{\frac{1}{m},p^{0}} ,\dots,\lra{\frac{1}{m},p^{m-1}})\right) \;\; ,\;\; 
p'' \in \conv\left(\Vsupp(\lra{\frac{1}{m},\psi\cdot p^{0}} ,\dots,\lra{\frac{1}{m},\psi\cdot p^{m-1}})\right)
$$
such that $(\Id\star \dots \star \Id)(p')=(\Id\star \dots \star \Id)(p'')$. By Proposition \ref{pro:VsuppDecompo} and 
since $\set{q^{i,j}:~0\leq i \leq m-1}$ is a closed set of vertices for every $0\leq j \leq m-1$, we have 
$$
\Vsupp\left(\lra{\frac{1}{m},p^{0}} ,\dots,\lra{\frac{1}{m},p^{m-1}}\right)= \bigsqcup_{j=0}^{m-1} \kappa_j
\left(\Vsupp(p^{j})\right)=\bigsqcup_{j=0}^{m-1} \set{\kappa_j(q^{0,j}),\dots,\kappa_j(q^{m-1,j})}
$$ 
So there is $\lambda_{ij} \in [0,1]$ where $\sum_{0\leq i,j \leq m-1}\lambda_{ij}=1$, such that
$$
p'=\sum_{0\leq i,j \leq m-1}\lambda_{ij} \kappa_{j}(q^{i,j})
$$
Thus $(\Id\star \dots \star \Id)(p')=\sum_{0\leq i,j \leq m-1}\lambda_{ij} q^{i,j}$. 
%
Similarly, and using Proposition \ref{pr:psiVsupp}, there is $\mu_{ij}\in [0,1]$ where 
$\sum_{0\leq i,j \leq m-1}\mu_{ij}=1$, such that 
$$
p''=\sum_{0\leq i,j \leq m-1}\mu_{ij} \kappa_{j}(\psi \cdot q^{i,j}) \;\; \text{and} \;\; 
(\Id\star \dots \star \Id)(p'')=\sum_{0\leq i,j \leq m-1}\mu_{ij} \psi \cdot q^{i,j}
$$
Since $(\Id\star \dots \star \Id)(p')|_L=(\Id\star \dots \star \Id)(p'')|_L$, we conclude that 
$$
\sum_{0\leq i,j \leq m-1}\lambda_{ij}  \delta^{\varphi^{i,j}}=\sum_{0\leq i,j \leq m-1}\mu_{ij} \psi^h \cdot \delta^{\varphi^{i,j}}
$$
Then $\lambda_{ij}=\mu_{ij}=\frac{1}{m^2}$ for every $0 \leq i,j \leq m-1$ according to Lemma \ref{lem:psih}. 
In that case, we have
$$
\begin{aligned}
p'=\sum_{j=0}^{m-1}\sum_{i=0}^{m-1}\frac{1}{m^2} \kappa_{j}(q^{i,j})&=\sum_{j=0}^{m-1}\frac{1}{m}\sum_{i=0}^{m-1}\frac{1}{m} \kappa_{j}(q^{i,j})\\
&=\sum_{j=0}^{m-1}\frac{1}{m}\kappa_{j}\left(\sum_{i=0}^{m-1}\frac{1}{m}q^{i,j}\right)\\
&=\sum_{j=0}^{m-1}\frac{1}{m}\kappa_{j}\left(p^j\right)\\
&=(\lra{\frac{1}{m},p^{0}} ,\dots,\lra{\frac{1}{m},p^{m-1}})
\end{aligned}
$$
Similarly, $p''=\left(\lra{\frac{1}{m},\psi \cdot p^{0}} ,\dots,\lra{\frac{1}{m},\psi \cdot p^{m-1}}\right)$.
So by Proposition \ref{thm:vertex on union} the simplicial distribution in (\ref{eq:FirstkindVert}) 
is a vertex in $\sDist(\Sigma X,\Delta_{\zz_m})$. Moreover, It is obviously contextual as it is not a deterministic distribution.
\end{proof}
Now, we apply Proposition \ref{pro:ConVertResult1} to prove that the joint probability distributions described in \cite[Equation (27)]{Nonlocalcorrela} is a contextual vertex. This joint probability distributions belong to the Bell scenario of three observers, each choose from two possible measurements with two outcomes.
%
\begin{example}
We define the deterministic distributions $\psi^{0,0}$, $\psi^{1,0}$, $\psi^{0,1}$ and $\psi^{1,1}$ on the CHSH scenario $(X,\Delta_{\zz_2})$ (Example \ref{ex:CHSHScenarioo}) to be 
%
$$
\psi^{i,j}_{\sigma_k}=
\begin{cases}
(i,i+j)   &   \text{if}\;\; k=1 \;\; \text{or} \;\; k=3   \\
(i+j,i)   &   \text{if}\;\; k=2 \;\; \text{or} \;\; k=4.
\end{cases}
$$
In addition, for $0\leq j \leq 1$ we define $p^{j}$ to be $\frac{1}{2}\delta^{\psi^{0,j}} + \frac{1}{2}\delta^{\psi^{1,j}}$. It turns out that
$$
p^j_{\sigma_k}=
\begin{cases}
p_{+}   &   \text{if}\;\; j=0   \\
p_{-}   &   \text{if}\;\; j=1
\end{cases}
$$
for every $1 \leq k \leq 4$ (see Equation (\ref{eq:p+p-})).
Note that $\Vsupp(p^0)=\set{\delta^{\psi^{0,0}},\delta^{\psi^{1,0}}}$ and $\Vsupp(p^1)=\set{\delta^{\psi^{0,1}},\delta^{\psi^{1,1}}}$. This means
that $\set{\delta^{\psi^{0,0}},\delta^{\psi^{1,0}}}$ and $\set{\delta^{\psi^{0,1}},\delta^{\psi^{1,1}}}$ are closed sets of vertices. 
Let $L$ be the line that generated by $\sigma_1$, $\sigma_2$, and 
$\sigma_3$. For every $0 \leq i,j\leq 1$ we define $\varphi^{i,j}$ to be $\psi^{i,j}|_{L}$.
One can check that $\varphi^{0,0}$, $\varphi^{1,0}$, $\varphi^{0,1}$, and $\varphi^{1,1}$ form a complete collection of deterministic distributions in the sense of Definition \ref{def:ComplColle} where 
$A=\begin{pmatrix} 
	1 & 1  \\
 1 & 0 \\
	\end{pmatrix}$ and $B=\begin{pmatrix} 
	1 & 0  \\
 1 & 1 \\
\end{pmatrix}$ 
and $h=1$.
We define $\psi \in \catsSet(X,\Delta_{\zz_2})$ to be $\psi_{\sigma_1}=(0,0)$, $\psi_{\sigma_2}=(0,1)$, $\psi_{\sigma_3}=(1,1)$, and $\psi_{\sigma_4}=(1,0)$. So $\psi|_{L}=\psi^1$ in terms of Lemma \ref{lem:psih}. 
Finally, both $\sum_{0\leq i,j \leq 1} \frac{1}{4} \delta^{\psi^{i,j}}$ and $\psi \cdot \sum_{0\leq i,j \leq 1} \frac{1}{4} \delta^{\psi^{i,j}}$ have 
the uniform distribution (Example \ref{ex:Uniform}) on every edge $\sigma_k$.
Therefore, by Proposition \ref{pro:ConVertResult1} the simplicial distribution 
$(\lra{\frac{1}{2},p^{0}},\lra{\frac{1}{2},p^{1}};
\lra{\frac{1}{2},\psi\cdot p^{0}},\lra{\frac{1}{2},\psi\cdot p^{1}})$ is a contextual vertex in $\sDist(\Sigma X,\Delta_{\zz_2})$. In fact, this vertex
is a (simplicial) version of the example that given in \cite[Equation (27)]{Nonlocalcorrela} for a class of 
so-called three-way non-local vertices.
\end{example}
\begin{example}
Given a complete collection of deterministic distributions as in Definition \ref{def:ComplColle} such that $\set{\delta^{\varphi^{i,j}}:~0\leq i \leq m-1}$ is closed set of vertices for every $0\leq j \leq m-1$. By Lemma \ref{lem:psih} and Proposition \ref{pro:ConVertResult1} we get that the following simplicial distribution
$$
\left(\lra{\frac{1}{m},\sum_{i=0}^{m-1}\frac{1}{m}\delta^{\varphi^{i,0}}} ,\dots,\lra{\frac{1}{m},
\sum_{i=0}^{m-1}\frac{1}{m}\delta^{\varphi^{i,m-1}}};\lra{\frac{1}{m},\psi^h \cdot\sum_{i=0}^{m-1}\frac{1}{m}\delta^{\varphi^{i,0}}},\dots,\lra{\frac{1}{m},\psi^h \cdot\sum_{i=0}^{m-1}\frac{1}{m}\delta^{\varphi^{i,m-1}}}\right)
$$
is a contextual vertex in $\sDist(\Sigma L^{(3)},\Delta_{\zz_m})$.
\end{example}
\begin{remark}
Consider the complete collection of deterministic distributions given in Example \ref{ex:compexam}, 
the set $\set{\delta^{\varphi^{i,j}}:~0\leq i \leq m-1}$ is a closed set of vertices for every $0\leq j \leq m-1$. This is not always the case, for example, if we define $\varphi^{0,0}$, $\varphi^{1,0}$, $\varphi^{0,1}$ and $\varphi^{1,1}$ in $\catsSet(L^{(3)},\Delta_{\zz_2})$ to be 
$$
\varphi^{i,j}_{\sigma_k}=
\begin{cases}
(i,j)   &   \text{if}\;\; k=1\;\; \text{or} \;\; k=3   \\
(j,i)   &   \text{if}\;\; k=2.
\end{cases}
$$
We obtain a complete collection of deterministic distributions, but 
$\set{\delta^{\varphi^{0,0}}, \delta^{\varphi^{1,0}}}$ is not closed since 
$\delta^\varphi \in \Vsupp(\frac{1}{2}\delta^{\varphi^{0,0}}+ \frac{1}{2}\delta^{\varphi^{1,0}})$ where 
$\varphi_{\sigma_1}=(0,0)$, $\varphi_{\sigma_2}=(0,1)$, and  $\varphi_{\sigma_3}=(1,0)$.
\end{remark}

\subsection{Complete collections of average distributions}\label{subsec:Finalll}

In this section, we write $L^{(2)}$ for the line with two edges $\sigma_1$ and $\sigma_2$ (see Definition \ref{def:lineee}). For simplicity, we assume that $d_{0}(\sigma_1)=d_1(\sigma_2)$.   
\begin{defn}\label{def:compIX}
A \emph{complete collection of average distributions} is 
a set $\set{q^{j}:~0\leq j \leq m-1}$ of simplicial distributions on $(L^{(2)},\Delta_{\zz_m})$ such that $\set{q_{\sigma_k}^{j}:~ 0\leq j \leq m-1}=\set{I,S,\dots,S^{m-1}}$ for every $1\leq k \leq 2$ (see Example \ref{def:RRRRR}).   
\end{defn}
Note that $D(d_0)(S^j)(a)=D(d_1)(S^j)(a)=\frac{1}{m}$ for every $a \in \zz_m$. So we have $(m!)^2$ complete collections of average distributions on $(L^{(2)},\Delta_{\zz_m})$.
\begin{example}\label{ex:compIXexam}
By setting  
$
q^{j}_{\sigma_k}=S^{j}
$
for every $0\leq j \leq m-1$ and $1\leq k \leq 2$, we obtain a complete collection of average distributions $\set{q^0,\dots,q^{m-1}}$. 
\end{example}
\begin{lemma}\label{lem:XIVert}
Given a complete collection of average distributions $\set{q^{j}:~0\leq j \leq m-1}$, and let 
$\psi \in \catsSet(L^{(2)},\Delta_{\zz_m})$ be defined by setting $\psi_{\sigma_1}=(0,0)$ and $\psi_{\sigma_2}=(0,1)$. If for $\lambda_{j},\mu_{j} \in [0,1]$ with $\sum_{j=0}^{m-1}{\lambda_{j}}=\sum_{j=0}^{m-1}{\mu_{j}}=1$ we have
\begin{equation}\label{eq:XIVert}
\sum_{j=0}^{m-1}\lambda_{j}q^{j}=\sum_{j=0}^{m-1}\mu_{j}\psi \cdot q^{j}   
\end{equation}
then 
$\lambda_{j}=\mu_{j}=\frac{1}{m}$ for every $0 \leq j \leq m-1$.
\end{lemma}
\begin{proof}
For simplicity, we will give a proof for the case that the complete collection of average distributions is the one given in Example \ref{ex:compIXexam}. 
By substituting $\sigma_1$ in Equation (\ref{eq:XIVert}) we get that for every $i \in \zz_m$ we have 
$$
\begin{aligned}
\lambda_i\frac{1}{m}&=\lambda_i S^i(0,i)=\sum_{j=0}^{m-1}\lambda_{j}S^j(0,i)\\
&=\sum_{j=0}^{m-1}\lambda_{j}q^{j}_{\sigma_1}(0,i) =\sum_{j=0}^{m-1}\mu_{j}(\delta^{\psi_{\sigma_1}} \ast q^{j}_{\sigma_1})(0,i) \\
&= \sum_{j=0}^{m-1} \mu_{j} (\delta^{(0,0)} \ast S^j)(0,i) =\sum_{j=0}^{m-1}\mu_{j}S^j(0,i) \\
&=\mu_i S^i(0,i)=\mu_i\frac{1}{m}
\end{aligned}
$$
see Equations (\ref{eq:Sjjj}) and (\ref{eq:deltaR}). We conclude that $\lambda_i=\mu_i$ for every $0 \leq i \leq m-1$. On the other hand, by substituting $\sigma_2$ in Equation (\ref{eq:XIVert}) we get that for every $i \in \zz_m$ we have 
$$
\begin{aligned}
\lambda_i\frac{1}{m}&=\lambda_i S^i(0,i)=\sum_{j=0}^{m-1}\lambda_{j}S^j(0,i)\\
&=\sum_{j=0}^{m-1}\lambda_{j}q^{j}_{\sigma_2}(0,i) =\sum_{j=0}^{m-1}\mu_{j}(\delta^{\psi_{\sigma_2}} \ast q^{j}_{\sigma_2})(0,i) \\
&= \sum_{j=0}^{m-1} \mu_{j} (\delta^{(0,1)} \ast S^j)(0,i) =\sum_{j=0}^{m-1}\mu_{j}S^{j+1}(0,i) \\
&=\mu_{i-1}S^i(0,i)=\mu_{i-1}\frac{1}{m}
\end{aligned}
$$
We conclude that $\lambda_{i+1}=\mu_i$ for every $0 \leq i \leq m-1$. Thus $\lambda_0=\mu_0=\lambda_1=\mu_1=\lambda_2=\dots=\lambda_{m-1}=\mu_{m-1}$. The assumption $\sum_{j=0}^{m-1}{\lambda_{j}}=1$ implies that $\lambda_{i}=\mu_{i}=\frac{1}{m}$ for every $0 \leq i \leq m-1$.     

%
%
\end{proof}
\begin{prop}\label{pro:ConVertResult2}
Let $X$ be a finitely generated, connected simplicial set. Given a set of vertices $\set{p^{j}:~0\leq j \leq m-1}$ in $\sDist(X,\Delta_{\zz_m})$, and a line $L=L^{(2)} \subseteq X$ generated by $\sigma_1$, $\sigma_2$, such that 
\begin{itemize}
    \item $\set{p^{j}|_{L}:~0\leq j \leq m-1}$ is a complete collection of average distributions as 
    in Definition \ref{def:compIX}.
    \item There is $\psi \in \catsSet(X,\Delta_{\zz_m})$ such that 
    $\psi_{\sigma_1}=(0,0)$, $\psi_{\sigma_2}=(0,1)$ and
    $\sum_{j=0}^{m-1} \frac{1}{m} p^{j}=\psi \cdot \sum_{j=0}^{m-1} \frac{1}{m} p^{j}$. 
\end{itemize} 
Then the simplicial distribution 
\begin{equation}\label{eq:SecondkindVert}   
\left(\lra{\frac{1}{m},p^{0}} ,\dots,\lra{\frac{1}{m},p^{m-1}};
\lra{\frac{1}{m},\psi\cdot p^{0}} ,\dots,\lra{\frac{1}{m},\psi\cdot p^{m-1}}\right)
\end{equation}
is a contextual vertex in $\sDist(\Sigma X ,\Delta_{\zz_m})$.
\end{prop}
\begin{proof}
The simplicial distribution in (\ref{eq:SecondkindVert}) is well defined since  $\sum_{j=0}^{m-1} \frac{1}{m} p^{j}=\psi \cdot \sum_{j=0}^{m-1} \frac{1}{m} p^{j}$. Now, given simplicial distributions 
$$
p' \in \conv\left(\Vsupp(\lra{\frac{1}{m},p^{0}} ,\dots,\lra{\frac{1}{m},p^{m-1}})\right) \;\; ,\;\; 
p'' \in \conv\left(\Vsupp(\lra{\frac{1}{m},\psi\cdot p^{0}} ,\dots,\lra{\frac{1}{m},\psi\cdot p^{m-1}})\right)
$$
such that $(\Id\star \dots \star \Id)(p')=(\Id\star \dots \star \Id)(p'')$. By Proposition \ref{pro:VsuppDecompo} we have 
$$
\Vsupp\left(\lra{\frac{1}{m},p^{0}} ,\dots,\lra{\frac{1}{m},p^{m-1}}\right)= \set{\kappa_j(p^j):~0\leq j\leq m-1}
$$ 
So there is scalars $\lambda_{j} \in [0,1]$ such that 
$\sum_{j=0}^{m-1}\lambda_{j}=1$ and
$
p'=\sum_{j=0}^{m-1}\lambda_{j} \kappa_{j}(p^{j})
$.
Thus we have $(\Id\star \dots \star \Id)(p')=\sum_{j=0}^{m-1}\lambda_{j} p^{j}$. 
%
Similarly, there is scalars $\mu_{j}\in [0,1]$ where 
$\sum_{j=0}^{m-1}\mu_{j}=1$, such that 
$$
p''=\sum_{j=0}^{m-1}\mu_{j} \kappa_{j}(\psi \cdot p^{j}) \;\; \text{and} \;\; 
(\Id\star \dots \star \Id)(p'')=\sum_{j=0}^{m-1}\mu_{j} \psi \cdot p^{j}
$$
We conclude that 
$$
\sum_{j=0}^{m-1}\lambda_{j}  p^j|_L=(\Id\star \dots \star \Id)(p')|_L=(\Id\star \dots \star \Id)(p'')|_L=
\sum_{j=0}^{m-1}\mu_{j} \psi|_{L}\cdot p^j|_L
$$
by Lemma \ref{lem:XIVert} we get that $\lambda_{j}=\mu_{j}=\frac{1}{m}$ for every $0 \leq j \leq m-1$. This determines 
$(p';p'')$ uniquely. Therefore, by Proposition \ref{thm:vertex on union} the simplicial distribution in (\ref{eq:SecondkindVert}) 
is a contextual vertex.
\end{proof}
%
%
%
%
%
\begin{example}
Let $p^0$ and $p^1$ be PR boxes (see Example \ref{ex:PRBoxxx}) on the CHSH scenario $(X,\Delta_{\zz_2})$ defined by setting
$$
p^j_{\sigma_k} =\begin{cases}
p_{+} & \text{if}\;\; (k=1,2,4 \;\; \text{and} \;\; j=0) \;\; \text{or} \;\; (k=3 \;\; \text{and} \;\; j=1) \\
p_{-} & \text{if}\;\; (k=1,2,4 \;\; \text{and} \;\; j=1) \;\; \text{or} \;\; (k=3 \;\; \text{and} \;\; j=0) 
\end{cases}
$$
and define $\psi \in \catsSet(X,\Delta_{\zz_2})$ by setting $\psi_{\sigma_1}=(0,0)$, $\psi_{\sigma_2}=(0,1)$, 
$\psi_{\sigma_3}=(1,1)$, and $\psi_{\sigma_4}=(1,0)$. Let $L$ be the line generated by $\sigma_1$ and $\sigma_2$, then 
$\set{p^0|_{L},p^1|_L}$ is a complete collection of average distributions. Furthermore, the simplicial distributions $\frac{1}{2}p^0+\frac{1}{2}p^1$ and 
$\psi\cdot(\frac{1}{2}p^0+\frac{1}{2}p^1)$ are uniform (Example \ref{ex:Uniform}) on every generated edge $\sigma_k$. Therefore, by Proposition 
\ref{pro:ConVertResult2} we get that 
$$  
\left(\lra{\frac{1}{2},p^{0}},\lra{\frac{1}{2},p^{1}};
\lra{\frac{1}{2},\psi\cdot p^{0}} ,\lra{\frac{1}{2},\psi\cdot p^{1}}\right)
$$
is a contextual vertex in $\sDist(\Sigma X ,\Delta_{\zz_2})$. This vertex represents the one that given in \cite[Equation (28)]{Nonlocalcorrela}.
\end{example}


\begin{thebibliography}{DKSm}





\bibitem[1]{RaussContex} R~.Raussendorf, ``Contextuality in measurement-based quantum computation",\hsm 
\textit{Physical Review A}, vol. 88, no. 2, p. 022322, 2013.

\bibitem[2]{HowarContex} M.~Howard, J.~Wallman, V.~Veitch,  \& J.~Emerson, ``Contextuality supplies the magic for
 quantum computation",\hsm \textit{Nature}, vol. 510, no. 7505, pp. 351-355, 2014.

\bibitem[3]{BravyiQuan} S.~Bravyi, D.~Gosset, \& R.~Konig, ``Quantum advantage with shallow circuits",\hsm \textit{Science},
 vol. 362, no. 6412, pp. 308-311, 2018.

\bibitem[4]{RauTherole}
 R.~Raussendorf, C.~Okay, M.~Zurel,  \& P.~Feldmann, ``The role of cohomology in quantum
 computation with magic states",\hsm \textit{Quantum}, vol. 7, p. 979, 2023

\bibitem[5]{Abrams1} S.~Abramsky  \& A.~Brandenburger, ``The sheaf-theoretic structure of non-locality and contextuality"\hsm \textit{New Journal of Physics}, vol. 13, no. 11, p. 113036, 2011. doi: 10.1088/1367-2630/13/11/113036.

\bibitem[6]{Abrams2} S.~Abramsky, S.~Mansfield  \& R.~Barbosa, 
``The cohomology of non-locality and contextuality",\hsm \textit{Electronic Proceedings in Theoretical Computer Science 95} (2012) 
- Proceedings 8th International Workshop on Quantum Physics and Logic (QPL 2011), Nijmegen, pp. 1-14.


\bibitem[7]{OkayTopo} C.~Okay, S.~Roberts, S.D.~Bartlett,  \& R.~Raussendorf, ``Topological proofs of contextuality
in quantum mechanics",\hsm \textit{Quantum Information \& Computation}, vol. 17, no. 13-14, pp. 1135-1166, 2017. doi: 10.26421{/}QIC17.13-14-5.

\bibitem[8]{OkayQuan}
C.~Okay, A.~Kharoof \& S.~Ipek,
``Simplicial quantum contextuality",\hsm \textit{Quantum} \textbf{7} (2023), p. 1009. 

\bibitem[9]{JardineHomotopy} P. G.~Goerss \& J. F.~Jardine, Simplicial homotopy theory.\hsm Springer Science \& Business Media,
2009.
 
\bibitem[10]{TopoloMeth}
A.~Kharoof, S.~Ipek \& C.~Okay, ``Topological methods for studying contextuality: N-cycle
scenarios and beyond",\hsm \textit{Entropy}, vol. 25, no. 8, 2023.

\bibitem[11]{HomVert}
A.~Kharoof \& C.~Okay, ``Homotopical characterization of strongly contextual simplicial distributions on cone spaces",\hsm \textit{Topology and its Applications} 352 (2024): 108956.

\bibitem[12]{OkayRank} C.~Okay, ``On the rank of two-dimensional simplicial distributions",\hsm
arXiv preprint arXiv:2312.15794 (2023).

\bibitem[13]{EXtrCycle}
A.~Kharoof, S.~Ipek, \& C.~Okay,
``Extremal simplicial distributions on cycle scenarios with arbitrary outcomes",\hsm arXiv
preprint arXiv:2406.19961 (2024).

\bibitem[14]{BellFirst} J. S.~Bell, ``On the Einstein Podolsky Rosen paradox”,\hsm \textit{Physics Physique Fizika}, 
vol. 1, pp. 195-200, Nov 1964. doi: 10.1103/PhysicsPhysiqueFizika.1.195.

\bibitem[15]{Fineee} A.~Fine, ``Hidden variables, joint probability, and the Bell inequalities",\hsm \textit{Physical Review Letters} 48.5 (1982): 291.

\bibitem[16]{PRBoxesss} S.~Popescu and D.~Rohrlich, ``Quantum nonlocality as an axiom",\hsm \textit{Foundations of Physics},
vol. 24, no. 3, pp. 379-385, 1994. doi: 10.1007/BF02058098.

\bibitem[17]{Nonlocalcorrela} J.~Barrett, N.~Linden, S.~Massar, S.~Pironio, S.~Popescu, and D.~Roberts, ``Nonlocal correlations
as an information-theoretic resource",\hsm \textit{Physical Review A}, vol. 71, p. 022101, Feb 2005. doi:
10.1103/PhysRevA.71.022101. arXiv: quant-ph/0404097.

\bibitem[18]{TodC}
H.~Toda,
``Composition methods in the homotopy groups of spheres",\hsm
\textit{Annals of Mathematics
Studies}, No. 49, Princeton University Press, Princeton, N.J., 1962. MR0143217.

\bibitem[19]{KharoH}
A.~Kharoof,
``Higher order Toda brackets",\hsm
\textit{J.\ Homotopy \& Rel.\ Str.} \textbf{16} (2021), pp.~451-486.

\bibitem[20]{FriedmansSett} G.~Friedman, ``An elementary illustrated introduction to simplicial sets",\hsm 
arXiv preprint arXiv:0809.4221 (2008).

\bibitem[21]{MaclaneCat} S.M.~Lane. Categories for the Working Mathematician,\hsm Springer New York,
1978.

\bibitem[22]{ConCat}
A.~Kharoof \& C.~Okay,  ``Simplicial distributions, convex categories and contextuality",\hsm arXiv
preprint arXiv:2211.00571 (2022).

\bibitem[23]{CHSHScennn} J.~Clauser, M.~Horne, A.~Shimony, \& R.~Holt, ``Proposed
experiment to test local hidden-variable theories",\hsm \textit{Physical review letters}, 23(15):880,
1969. doi: 10.1103/PhysRevLett.23.880.

\bibitem[24]{BundlePaper} 
R.~Barbosa, A.~Kharoof, \& C.~Okay,
``A bundle perspective on contextuality: Empirical models and simplicial distributions on bundle 
scenarios",\hsm arXiv
preprint arXiv:2308.06336 (2023).

\bibitem[25]{ContexFract} S.~Abramsky, R.~Barbosa, \& S.~Mansfield, ``Contextual fraction as a measure of contextuality",\hsm \textit{Physical review letters} 119.5 (2017): 050504.

\bibitem[26]{Foresat} M.~Froissart, ``Constructive generalization of Bell’s inequalities",\hsm \textit{Nuovo Cimento B;(Italy)},
vol. 64, no. 2, 1981. doi: 10.1007/BF02903286.

\bibitem[27]{JacobsConv} B.~Jacobs, ``Convexity, duality and effects",\hsm \textit{In IFIP Advances in Information and
Communication Technology}, pages 1-19. Springer Berlin Heidelberg, 2010.

\bibitem[28]{OperadConv} R.~Haderi, C.~Okay, and W.~Stern, ``The operadic theory of convexity",\hsm arXiv preprint arXiv:2403.18102 (2024).

\bibitem[29]{Cocompleteee} M.~Barr \& C.~Wells. Toposes, triples and theories,\hsm volume 278. Springer-Verlag New York, 1985.

\bibitem[30]{StateConv} B.~Jacobs, B.~Westerbaan, \& B.~Westerbaan, ``States of convex sets",\hsm \textit{Foundations of Software Science and Computation Structures}: 18th International Conference, FOSSACS 2015, Held as Part of the European Joint Conferences on Theory and Practice of Software, ETAPS 2015, London, UK, April 11-18, 2015, Proceedings 18. Springer Berlin Heidelberg, 2015.

\bibitem[31]{Riehlll} E.~Rieh. Category theory in context,\hsm Courier Dover Publications, 2017.

\bibitem[32]{Decle}
D.~Stevenson, ``D´ecalage and Kan’s simplicial loop group functor",\hsm \textit{Theory Appl. Categ.}, vol. 26,
pp. No. 28, 768-787, 2012.


\bibitem[33]{Decle2} 
D.~Stevenson, ``Classifying theory for simplicial parametrized groups",\hsm arXiv preprint
 arXiv:1203.2461 (2012).
     
%




%


%

\end{thebibliography}
\end{document}